\documentclass[11pt]{amsart}
\usepackage{amsmath, amsthm, amsfonts, amsbsy, amssymb, upref, enumerate, bigstrut,  color, mathtools, mathrsfs, float, bm, dsfont}
\usepackage{psfrag}
\usepackage{natbib}
\usepackage{longtable}

\usepackage[left=1.3in,top=1.3in,right=1.3in]{geometry}

\usepackage{scalefnt}

\usepackage{multirow}
\usepackage{ragged2e}
\usepackage{siunitx}
\usepackage{makecell}

\usepackage[justification=centering]{caption}

\usepackage{epsfig, graphicx}
\usepackage{ hyperref}

\newtheorem{theorem}{Theorem}[section]

\newtheorem{corollary}[theorem]{Corollary}
\newtheorem{proposition}[theorem]{Proposition}

\newtheorem{remark}[theorem]{Remark}
\newtheorem{Example}[theorem]{Example}

\begin{document}
\title[Discrete Log-symmetric Distributions]
{On a family of discrete log-symmetric distributions}
\author{H. Saulo, R.  Vila, L. Paiva and N. Balakrishnan}
\address{
    \newline
    Departamento de Estat\'istica -
Universidade de Bras\'ilia, 70910-900, DF, Brazil,
    \newline
    {\rm Helton Saulo }, \quad Email: \textup{\tt heltonsaulo@gmail.com }
    \newline
    {\rm Roberto  Vila}, \quad  Email: \textup{\tt rovig161@gmail.com }
    \newline
    {\rm Leonardo Paiva}, \quad Email: \textup{\tt  leopaiva23@hotmail.com}
 \newline
 \newline
    Department of Mathematics and Statistics, McMaster University, Hamilton, ON, Canada
   \newline
{\rm Narayanaswamy Balakrishnan}, \quad Email: \textup{\tt bala@mcmaster.ca}
}

\date{\today}


\begin{abstract}
The use of continuous probability distributions has been widespread in problems with purely discrete nature. In general, such distributions are not appropriate in this scenario. In this paper, we introduce a class of discrete and asymmetric distributions based on the family of continuous log-symmetric distributions. Some properties are discussed as well as estimation by the maximum likelihood method. A Monte Carlo simulation study is carried out to evaluate the performance of the estimators, and censored and uncensored data sets are used to illustrate the proposed methodology.
\end{abstract}

\maketitle

\section{Introduction}\label{sec:1}

Continuous log-symmetric distributions are of particular interest for describing strictly positive and asymmetric data with the possibility of outlier observations; see, for example, \cite{j:08}, \cite{vp:16a,vanegasp:16b}, \cite{saulo2017log}, \cite{Balakrishnan2017}, \cite{franciscosilvia2017} and \cite{venturaletal:19}, for some discussions and applications of log-symmetric models. A continuous random variable $Y$ follows a log-symmetric distribution if its probability density function (PDF) is given by
\begin{align}\label{eq:ft}
f_{Y}(y|\boldsymbol{\theta})
=
\dfrac{(Z_g)^{-1}}{\sqrt{\phi}\,y}\, g\big[a_{\boldsymbol{\theta}}^2(y)\big],
\quad
a_{\boldsymbol{\theta}}(y)
=
\log \Big({y\over\lambda}\Big)^{1/\sqrt{\phi}},
\quad y>0;
\end{align}
where
$\boldsymbol{\theta}= (\lambda, \phi)$, $\lambda>0$ is a scale parameter and also the median of $Y$, $\phi>0$ is a shape parameter associated with the skewness or relative dispersion, $Z_g=\int^{\infty}_{-\infty} g(w^2)  \,\textrm{d}w$ is the partition function, and the function $g$ is a density generating kernel such that $g(u)>0$ for $u>0$. The function $g$ is associated with an additional parameter $\xi$ (or vector $\bm\xi$). We use the notation $Y\sim\textrm{LS}(\boldsymbol{\theta},g)$. Note that if 
$g(u)$ in \eqref{eq:ft} is 
$\exp(-u/2)$; 
$[1+(u/\xi)]^{-(\xi+1)/{2}}$, $\xi>0$; 
$\exp[-u^{{1}/(1+\xi)}/2]$, $-1<{\xi}\leqslant{1}$; 
$\sqrt{\xi_2}\exp(-\xi_2 u/2)+[{(1-\xi_1)}/{\xi_1}]
\exp(-u/2)$, $0<\xi_1,\xi_2<1$;
$\cosh(u^{1/2})\exp[-({2}/{\xi^2})\sinh^2(u^{1/2})] $, $\xi>0$; or 
$\cosh(u^{1/2})\left[\xi_{2}\xi_{1}^2+4\sinh^2(u^{1/2})\right]^{-({\xi_{2}+1})/{2}} $, $\xi_{1},\xi_{2}>0$; we have the 
log-normal,
log-Student-$t$,
log-power-exponential,
log-contaminated-normal,
extended Birnbaum-Saunders or
extended Birnbaum-Saunders-$t$ distributions, respectively; see \cite{vp:16a} and \cite{venturaletal:19}. If $Y\sim\textrm{LS}(\boldsymbol{\theta},g)$, then the associated cumulative distribution function (CDF) is given by
$
F_{Y}(y|\boldsymbol{\theta})
=G\big[a_{\boldsymbol{\theta}}(y)\big],
$
where the
function $G:(-\infty,+\infty)\to [0,1]$ is defined as
\begin{align}\label{def-a-g} 
G(r)
=
(Z_g)^{-1}
{\int^{r}_{-\infty} g(z^2)  \,\textrm{d}z },
\quad -\infty<r<+\infty.
\end{align}
%
This mapping is easily seen to have the following properties:
\begin{itemize}
	\item[(a)] $G(0)=0.5$, $G(+\infty)=\lim_{r\to+\infty}G(r)=1$,
	$G(-\infty)=\lim_{r\to-\infty}G(r)=0$; 
	\item[(b)] $G(\cdot)$ is a continuous function and that $G(\cdot)$ is strictly monotonically increasing. Hence $G(\cdot)$ has an inverse function, denoted by $G^{\pmb{-1}}(\cdot)$;  
	\item[(c)] From Items (a) and (b), $G(\cdot)$ is a CDF; and  
	\item[(d)] $G^{\pmb{-1}}(1-p)=-G^{\pmb{-1}}(p)$ for $p\in(0,1)$ given.
\end{itemize}

Despite the huge use of log-symmetric distributions -- its most famous member is the log-normal model -- they are not appropriate in purely discrete contexts. For example, to model the number of cycles before failure of a equipment or the number of weeks to cure a patient, among others; see \cite{vns:19}. Moreover, despite useful, continuous log-symmetric models do not include the zero. In this paper, we define a discrete random variable associated to $Y$ in \eqref{eq:ft} as  
$
X=\lfloor Y\rfloor,
$
where $\lfloor y\rfloor$ denotes the largest integer less than or equal to $y$. In other words, we propose a class of discrete log-symmetric distributions. The proposed class incorporates every distribution belonging to the log-symmetric family, and it is useful for asymmetric and non-negative discrete data.

The rest of the paper proceeds as follows. In Section~\ref{sec:02}, we introduce the class of  discrete log-symmetric models. In Section \ref{sec:math}, we discuss some mathematical properties. In Section~\ref{sec:03}, estimation of the model parameters are approached via the maximum likelihood method for the censored and uncensored cases. In Section ~\ref{sec:04}, we carry out a simulation study to evaluate the performance of the estimators taking into account different censoring proportions. In Section~\ref{sec:05}, we illustrate the proposed methodology with two real data sets. Finally, in Section~\ref{sec:05}, we make some concluding remarks and discuss future work.

\section{Discrete log-symmetric distributions}\label{sec:02}
We say that a discrete random variable $X$, taking values in the set $\{0,1,\ldots\}$, follows a { discrete log-symmetric distribution} 
with parameter vector $\boldsymbol{\theta}= (\lambda, \phi)$, where 
$\lambda>0, \phi>0$,
denoted by $X\sim\textrm{LS}_{\rm d}(\boldsymbol{\theta},g)$, if its 
probability mass function (PMF) is given by
\begin{align}\label{relation}
p(x|\boldsymbol{\theta})=
G\big[a_{\boldsymbol{\theta}}(x+1)\big]
-
G\big[a_{\boldsymbol{\theta}}(x)\big], \quad  
x=0,1, \ldots,
\end{align}
where $a_{\boldsymbol{\theta}}(\cdot)$ and
$G(\cdot)$ are as in \eqref{eq:ft} and \eqref{def-a-g}, respectively. Note that $G\big[a_{\boldsymbol{\theta}}(0)\big]=G(-\infty)=0$ and
that $G\big[a_{\boldsymbol{\theta}}(+\infty)\big]=G(+\infty)=1$.
Given the density generating kernel $g$, defined below Item \eqref{eq:ft}, the parameters $\lambda$ and $\phi$ completely determine the PMF 
\eqref{relation} at $x=0.$
Since $G(\cdot)$ and $a_{\boldsymbol{\theta}}(\cdot)$ are strictly increasing functions, and 
\[
\lim_{n\to\infty}
\sum_{x=0}^{n}p(x|\boldsymbol{\theta})
=
\lim_{n\to\infty}
G\big[a_{\boldsymbol{\theta}}(n+1)\big]
=
G(+\infty)
=
1,
\]
it is clear that $p(x|\boldsymbol{\theta})$ is a PDF.

The CDF, reliability function (RF) and hazard rate (HR)
of the $\textrm{LS}_{\rm d}$ distribution, respectively, are given by
\begin{align*} 
& F(x|\boldsymbol{\theta})
=
1- R(x|\boldsymbol{\theta})
=
G\big[a_{\boldsymbol{\theta}}(\lfloor x\rfloor+1)\big], \quad 
x\geqslant 0;
\\[0,2cm]
& H(x|\boldsymbol{\theta})
=
{p(x|\boldsymbol{\theta})\over p(x|\boldsymbol{\theta})+R(x|\boldsymbol{\theta})}
=
{
	G\big[a_{\boldsymbol{\theta}}(x+1)\big]
	-
	G\big[a_{\boldsymbol{\theta}}(x)\big]
	\over 	
	1-G\big[a_{\boldsymbol{\theta}}(x)\big]
}, \quad 
x=0,1,\ldots.
\end{align*}

%
%
\section{Mathematical properties}\label{sec:math}

This section, if not explicitly mentioned otherwise, consists of mathematical properties valid for any discrete random variable $X$ with support $\{0,1,\ldots\}$. 

Let $(b_n)$ be a sequence of real numbers.
For technical reasons in the next result we  use the convention $\prod_{y=0}^{-1} b_y=1$.
The next result provides a characterization of the PMF  and RF 
of a discrete distribution in terms of the HR.
\begin{proposition}\label{chac-re}
	If $X$ is a discrete random variable 
	then, for each $x=0,1,2,\ldots,$
	\begin{itemize}
		\item[\rm (a)]
		$\displaystyle
		p(x|\boldsymbol{\theta})
		= 
		{H(x|\boldsymbol{\theta})\over 1- H(x|\boldsymbol{\theta})}\,
		\prod_{y=0}^{x-1}
		\big[1-H(y|\boldsymbol{\theta})\big];$ 
		\item[\rm (b)] 
		$\displaystyle
		R(x|\boldsymbol{\theta})
		= 
		\prod_{y=0}^{x-1}
		\big[1-H(y|\boldsymbol{\theta})\big];$
	\end{itemize}
	where $H(\cdot|\boldsymbol{\theta})$ is the HR.
\end{proposition}
\begin{proof}
	By using the identity
	$
	p(x|\boldsymbol{\theta})
	=
	R(x|\boldsymbol{\theta})-R(x+1|\boldsymbol{\theta})
	=
	\big[p(x|\boldsymbol{\theta})+R(x|\boldsymbol{\theta})\big] H(x|\boldsymbol{\theta}), \ x=0,1,\ldots,
	$
	we have 
	\[
	1=H(x|\boldsymbol{\theta})+{R(x|\boldsymbol{\theta})H(x|\boldsymbol{\theta})\over p(x|\boldsymbol{\theta})} , \quad x=0,1,2,\ldots.
	\]
	Since ${p(x|\boldsymbol{\theta})/ H(x|\boldsymbol{\theta})}
	=
	p(x|\boldsymbol{\theta})+R(x|\boldsymbol{\theta})
	=
	R(x-1|\boldsymbol{\theta}),
	$
	it follows that
	\[
	{R(x|\boldsymbol{\theta})\over R(x-1|\boldsymbol{\theta})}
	=
	1-H(x|\boldsymbol{\theta}),
	\quad 
	x=0,1,2,\ldots.
	\]
	Exchanging $x$ for $y$ in the above identity and then
	multiplying from $y=0$ to $y=x-1$, we get
	\[
	R(x|\boldsymbol{\theta})
	= 
	\prod_{y=0}^{x-1}
	{R(y|\boldsymbol{\theta})\over R(y-1|\boldsymbol{\theta})}
	= \textstyle
	\prod_{y=0}^{x-1}
	\big[1-H(y|\boldsymbol{\theta})\big],
	\quad 
	x=0,1,2,\ldots,
	\]
	verifying the identity for $R(x|\boldsymbol{\theta})$. On the other hand,
	combining the above identity with the definition of HR, the identity for
	$p(x|\boldsymbol{\theta})$ follows.
\end{proof}

\subsection{Moments and variance}
\begin{theorem}\label{moments}
	If $X$ is a discrete random variable possessing all the higher-order moments,
	then
	\begin{itemize}
		\item[\rm (a)]
		$\displaystyle
		\mathbb{E}(X^r) =  \sum_{x=0}^{\infty}\big[(x+1)^r-x^r\big] R(x|\boldsymbol{\theta});$ 
		\item[\rm (b)] 
		$\displaystyle
		\mathbb{E}(X^r)
		=
		\sum_{x=0}^{\infty}\sum_{k=0}^{r}
		\sum_{i=0}^{r-k}
		\binom{r-k}{i}
		x^{k+i}\, R(x|\boldsymbol{\theta});$ 
		\item[\rm (c)]
		$\displaystyle
		{\rm Var}(X)=2\sum_{x=0}^{\infty} x R(x|\boldsymbol{\theta})
		+
		\sum_{x=0}^{\infty} R(x|\boldsymbol{\theta})
		\bigg[
		1-\sum_{x=0}^{\infty} R(x|\boldsymbol{\theta})
		\bigg];
		$
	\end{itemize}
	where $R(\cdot|\boldsymbol{\theta})$ is the RF.
\end{theorem}
\begin{proof}
	In order to prove Item (a),
	using the telescopic series $\sum_{x=0}^{i-1} \big[(x+1)^r-x^r\big]=i^r$, it follows that
	\begin{align*}
	\mathbb{E}(X^r)
	=
	\sum_{i=0}^{\infty} \sum_{x=0}^{\infty} \mathds{1}_{\{x<i\}}  \big[(x+1)^r-x^r\big] \,
	p(i|\boldsymbol{\theta})
	=
	\sum_{x=0}^{\infty}  \big[(x+1)^r-x^r\big]  \sum_{i=0}^{\infty} \mathds{1}_{\{i>x\}} \,
	p(i|\boldsymbol{\theta}),
	\end{align*}
	where in the second equality we exchange the orders of the summations because the following series
	\begin{align*}
	\sum_{x=0}^{\infty} \mathds{1}_{\{x<i\}}   \big|(x+1)^r-x^r\big| \,
	p(i|\boldsymbol{\theta})
	= 
	\sum_{x=0}^{\infty} \mathds{1}_{\{x<i\}} \cdot \big[(x+1)^r-x^r\big] \,
	p(i|\boldsymbol{\theta})
	=
	i^r p(i|\boldsymbol{\theta})\eqqcolon M_i 
	\end{align*}
	is finite for each $i=0,1,\ldots$,
	and, by hypothesis, the expectation $\mathbb{E}(X^r)=\sum_{i=0}^{\infty}M_i$ always exists.
	This proves the first item.
	The second item follows by combining the expression for $\mathbb{E}X^r$ given in the first item with
	the polynomial identity $a^n-b^n = (a-b) \sum_{k=0}^{r} a^{r-k}b^k$ and with the binomial expansion.
	The proof of the third item immediately follows from Item (a).
	Thus, the proof is complete.
\end{proof}

\subsection{The $p$-quantile}
\begin{theorem}\label{quantile}
	Let 
	$X=\lfloor Y\rfloor$ be a discrete random variable obtained from a positive continuous random variable $Y$ with CDF $F_Y(\cdot|\boldsymbol{\theta})$.
	Given $p\in(0,1)$, let $Q_p=F_Y^{\pmb{-1}}(p|\boldsymbol{\theta})$ be the $p$-quantile for $Y$. The following statements are valid:
	\begin{itemize}
		\item[\rm (a)] 
		If $Q_p>0$ is a natural number, then $Q_p-1$ is the $p$-quantile for $X$;
		\item[\rm (b)] 
		If $Q_p>0$ is not a natural number, then 
		all $y\in\big[\lfloor Q_p\rfloor,\lfloor Q_p\rfloor+1\big)$ is a $p$-quantile for $X$.
	\end{itemize}	
\end{theorem}  
\begin{proof}
	Given $p\in(0,1)$, assume that $Q_p=F_Y^{\pmb{-1}}(p|\boldsymbol{\theta})$.
	By using the relations, for all $x>0$,
	\begin{align}
	& \displaystyle
	F_X(x^-|\boldsymbol{\theta})=F_Y(\lfloor x\rfloor|\boldsymbol{\theta});
	\label{id-1}
	\\
	& \displaystyle
	\lfloor x\rfloor\leqslant x < \lfloor x\rfloor+1;
	\label{id-2}
	\end{align} 
	where $F_X(x^-|\boldsymbol{\theta})=\lim_{\delta\to 0} F_X(x-\delta|\boldsymbol{\theta})$ for all $\delta>0$,
	we have	
	\begin{align*}
	F_X\big[(Q_p-1)^-|\boldsymbol{\theta}\big]
	\leqslant
	F_X\big(Q_p-1|\boldsymbol{\theta})
	=
	F_X\big(Q_p^-|\boldsymbol{\theta})
	\stackrel{\eqref{id-1}}{=}
	F_Y(Q_p|\boldsymbol{\theta})
	=
	p,
	\end{align*}
	whenever $Q_p>0$ is a natural number. Then, by definition of  $p$-quantile for a discrete random variable,  the statement in Item (a) follows.
	
	Already, when $Q_p>0$ is not a natural number, from \eqref{id-1} and \eqref{id-2} we have
	\begin{align*}
	& F_X(Q_p^-|\boldsymbol{\theta})
	\stackrel{\eqref{id-1}}{=}
	F_Y(\lfloor Q_p\rfloor|\boldsymbol{\theta})
	\stackrel{\eqref{id-2}}{\leqslant}
	F_Y(Q_p|\boldsymbol{\theta})
	=
	p;
	\\ 
	& F_X(Q_p|\boldsymbol{\theta})
	=
	F_X\big[(Q_p+1)^-|\boldsymbol{\theta}\big]
	\stackrel{\eqref{id-1}}{=}
	F_Y(\lfloor Q_p\rfloor+1|\boldsymbol{\theta})
	\stackrel{\eqref{id-2}}{\geqslant}
	F_Y(Q_p|\boldsymbol{\theta})
	=
	p.
	\end{align*}
	Therefore, $F_X(Q_p^-|\boldsymbol{\theta})\leqslant p\leqslant F_X(Q_p|\boldsymbol{\theta})$. Hence, the $p$-quantile for $X$ can be represented by any value in the interval $\big[\lfloor Q_p\rfloor,\lfloor Q_p\rfloor+1\big)$, and the proof of Item (b) follows. This completes the proof.
	%
	%
	%
	%
	%
\end{proof}

The following two results are applied exclusively to random variables with discrete log-symmetric distribution.
\begin{proposition}\label{prop-med}
	Let $X$ be a random variable with $\textrm{LS}_{\rm d}$ distribution. The following statements hold:
	\begin{itemize}
		\item[\rm (a)] 
		If $\lambda$ is a natural number, then $\lambda-1$ is the median for $X$;
		\item[\rm (b)] 
		If $\lambda$ is not a natural number, then the median of the distribution of $X$ can be represented by any value in the set
		$\big[\lfloor \lambda\rfloor,\lfloor \lambda\rfloor+1\big)$.		
	\end{itemize}
\end{proposition}
\begin{proof}
	Let $X\sim\textrm{LS}_{\rm d}(\boldsymbol{\theta},g)$.
	Since $G(\cdot)$ and $a_{\boldsymbol{\theta}}(\cdot)$ are strictly increasing functions, the function $G\big[a_{\boldsymbol{\theta}}(\cdot)\big]$ is a strictly increasing CDF corresponding to some continuous random variable $Y$ with log-symmetric distribution $\textrm{LS}(\boldsymbol{\theta},g)$.
	Furthermore, note that the median $Q_{0.5}$ for $Y$ can be written as 
	\begin{align*}
	Q_{0.5}
	=
	(G\circ a_{\boldsymbol{\theta}})^{\pmb{-1}}(0.5)
	=
	a^{\pmb{-1}}_{\boldsymbol{\theta}}\big[G^{\pmb{-1}}(0.5)\big]
	=\lambda\exp\left[\sqrt{\phi}\,G^{\pmb{-1}}(0.5)\right]
	=
	\lambda,
	\end{align*}
	where in the last equality we use that $G^{\boldsymbol{-1}}({0.5})=0$; see
	Item (a) below Item \eqref{def-a-g} in Section \ref{sec:1}.
	Then, by Theorem \ref{quantile}, the proof of Items (a) and (b) follows.
\end{proof}

Let $X\sim\textrm{LS}_{\rm d}(\boldsymbol{\theta},g)$. For given $p\in(0,1)$, let $Q_{{\rm d}; p}$ be a $p$-quantile for $X$. Let us define
\begin{align*}
& \text{Dispersion:}  \quad \zeta=Q_{{\rm d};0.75}-Q_{{\rm d};0.25}, \quad 0<\zeta<\infty;
\\
& \text{Relative dispersion:} \quad \varpi={\zeta\over \zeta+2
	Q_{{\rm d};0.25}}, \quad 0<\varpi<1;
\\
& \text{Skewness:} \quad \varkappa(p)={Q_{{\rm d};p}+Q_{{\rm d};1-p}-2
	Q_{{\rm d};0.5}\over Q_{{\rm d};1-p}+Q_{{\rm d};p}}, \quad 0<\varkappa(p)<1, \ 0<p<0.5;
\\
& \text{Kurtosis:} \quad \varsigma={Q_{{\rm d};7/8}-Q_{{\rm d};5/8}+
	Q_{{\rm d};3/8}-Q_{{\rm d};1/8}\over Q_{{\rm d};6/8}-Q_{{\rm d};2/8}}, \quad 
0\leqslant\varsigma<\infty.
\end{align*}
The relative dispersion, skewness and kurtosis have appeared in \cite{zwko:00}, \cite{hinkley:75} and \cite{moors:88}, respectively.
\begin{proposition}
	Given $p\in(0,1)$,
	let $X\sim\textrm{LS}_{\rm d}(\boldsymbol{\theta},g)$ and  let
	$Q_p$ be  the $p$-quantile of the corresponding 
	continuous log-symmetric distribution. If $Q_p$ is a natural number,  then
	\begin{itemize}
		\item[\rm (a)] 
		$\zeta=2\lambda \sinh\big[\sqrt{\phi}\, G^{\pmb{-1}}(0.75)\big];$ 
		\item[\rm (b)] 
		$\varpi=\big\{{\rm cotanh}\big[\sqrt{\phi}\, G^{\pmb{-1}}(0.75)\big] - {\rm cosech}\left[\sqrt{\phi}\, G^{\pmb{-1}}(0.75)\big]\right\}^{-1};$ 
		\item[\rm (c)] 
		$\varkappa(p)=\lambda;$
		\item[\rm (d)] $\displaystyle
		\varsigma=
		{\sinh\big[\sqrt{\phi}\, G^{\pmb{-1}}(7/8)\big]- 
			\sinh\big[\sqrt{\phi}\, G^{\pmb{-1}}(5/8)\big]\over 
			\sinh\big[\sqrt{\phi}\, G^{\pmb{-1}}(6/8)\big]};$
	\end{itemize}
	where $G(\cdot)$ was defined in \eqref{def-a-g}.
\end{proposition}
\begin{proof}
	Since $Q_p$ is a natural number,
	by Theorem \ref{quantile}, $Q_{{\rm d};p}=Q_p-1$ is a $p$-quantile for $X$, where $Q_p=\lambda\exp\big[\sqrt{\phi}\,G^{\pmb{-1}}(p)\big]$. By using the identity $G^{\pmb{-1}}(1-p)=-G^{\pmb{-1}}(p)$ (see
	Item (d) below Item \eqref{def-a-g} in Section \ref{sec:1}), from Proposition \ref{prop-med} and a simple algebraic computation, the proof of statements in Items (a)-(d) follows.
\end{proof}

\subsection{Shape properties}
The next result shows that the shape of a discrete log-symmetric distribution depends on choice of density generating kernel and on the distance between the modes of the corresponding continuous log-symmetric distribution.
\begin{theorem}\label{teo-shapes}
	Let $g$ be a density generating kernel so that the corresponding continuous log-symmetric distribution, of a  random variable $Y$, is bimodal. Then the discrete log-symmetric distribution of $X=\lfloor Y\rfloor$ has the following shapes:
	\begin{itemize}
		\item[\rm (a)]  It is bimodal, whenever the distance between the modes is big enough;   
		\item[\rm (b)]  It is unimodal, whenever the distance between the modes is small enough.
	\end{itemize}
\end{theorem}
\begin{proof}
	Since the proof of Item (b) follows the same analysis and steps as the first item, we are concerned with proving only Item (a). 
	
	Let $f_{Y}(y|\boldsymbol{\theta})$, $t>0$, be the bimodal PDF of the continuous random variable $Y\sim\textrm{LS}(\boldsymbol{\theta},g)$, where the distance between their modes, denoted by $y_0>0$ and $y_\epsilon=y_0+\epsilon$, is big enough ($\epsilon>6$). From bimodality property, there is $y_*\in(y_0, y_{\epsilon})$ such that the following inequalities hold:
	\begin{align}
	f_{Y}(y\vert\boldsymbol{\theta}) 
	\geqslant 
	f_{Y}(y-1|\boldsymbol{\theta})
	\quad \text{for all} \ y\leqslant y_0  \ \mbox{and} \  
	y_*\leqslant y \leqslant y_{\epsilon};   \label{relat-1}
	\\[0,2cm]
	f_{Y}(y\vert\boldsymbol{\theta}) 
	\geqslant 
	f_{Y}(y+1\vert\boldsymbol{\theta})
	\quad \mbox{for all} \ y_* \geqslant y\geqslant y_0  
	\ \mbox{and} \  y \geqslant y_{\epsilon}.   \label{relat-2}
	\end{align}
	If $x$ is a natural number such that $x\leqslant \lfloor y_0 \rfloor -1$ and
	$\lfloor y_*\rfloor +1\leqslant x \leqslant \lfloor y_{\epsilon}\rfloor-1$, from above inequalities, we have
	\begin{align*}
	p(x|\boldsymbol{\theta})
	=
	\int_{x}^{x+1}
	f_{Y}(y|\boldsymbol{\theta}) \, {\rm d} y
	\stackrel{\eqref{relat-1}}{\geqslant}
	\int_{x}^{x+1}
	f_{Y}(y-1|\boldsymbol{\theta}) \, {\rm d} y
	=
	p(x-1|\boldsymbol{\theta}).
	\end{align*}
	Already, if $x$ is a natural number such that
	$\lfloor y_*\rfloor-1\geqslant x\geqslant \lfloor y_0 \rfloor+1$
	and $x \geqslant \lfloor y_{\epsilon}\rfloor+1$, we have
	\begin{align*}
	p(x|\boldsymbol{\theta})
	=
	\int_{x}^{x+1}
	f_{Y}(y|\boldsymbol{\theta}) \, {\rm d} y
	\stackrel{\eqref{relat-2}}{\geqslant}
	\int_{x}^{x+1}
	f_{Y}(y+1|\boldsymbol{\theta}) \, {\rm d} y
	=
	p(x+1|\boldsymbol{\theta}).
	\end{align*}
	In other words, we have the following 
	\begin{align}\label{ineqs-1}
	\begin{array}{lllll}
	&p(0|\boldsymbol{\theta})
	\leqslant
	p(1|\boldsymbol{\theta})
	\leqslant
	\cdots
	\leqslant
	p(\lfloor y_0 \rfloor-1|\boldsymbol{\theta});
	\\
	&p(\lfloor y_0 \rfloor+1|\boldsymbol{\theta})
	\geqslant
	\cdots
	\geqslant
	p(\lfloor y_*\rfloor-2|\boldsymbol{\theta})
	\geqslant
	p(\lfloor y_*\rfloor-1|\boldsymbol{\theta});
	\end{array}
	\\[0,2cm]
	\label{ineqs-2}
	\begin{array}{lllll}
	&p(\lfloor y_*\rfloor +1|\boldsymbol{\theta})
	\leqslant
	p(\lfloor y_*\rfloor +2|\boldsymbol{\theta})
	\leqslant
	\cdots
	\leqslant
	p(\lfloor y_{\epsilon}\rfloor-1|\boldsymbol{\theta});
	\\
	&
	p(\lfloor y_{\epsilon}\rfloor+1|\boldsymbol{\theta})
	\geqslant
	p(\lfloor y_{\epsilon}\rfloor+2|\boldsymbol{\theta})
	\geqslant
	\cdots.
	\end{array}
	\end{align}
	From monotonicities \eqref{ineqs-1} and \eqref{ineqs-2}, one can guarantee the bimodality property of the discrete log-symmetric distribution.
	By using \eqref{ineqs-1}, we show how to obtain only one of the modes, since the other one can be obtained following a similar path. Indeed,  by \eqref{ineqs-1} it remains to relate the probabilities 
	$p(\lfloor y_0 \rfloor-1|\boldsymbol{\theta})$,
	$p(\lfloor y_0 \rfloor|\boldsymbol{\theta})$
	and  
	$p(\lfloor y_0 \rfloor+1|\boldsymbol{\theta})$
	to find the the first mode, denoted by $x_0$, of the discrete log-symmetric distribution. A simple observation shows that $x_0$ is given by
	\begin{align*}
	x_0=
	\begin{cases}
	\lfloor y_0 \rfloor+1, & \text{if} \
	p(\lfloor y_0 \rfloor-1|\boldsymbol{\theta})\leqslant
	p(\lfloor y_0 \rfloor|\boldsymbol{\theta})<
	p(\lfloor y_0 \rfloor+1|\boldsymbol{\theta}),
	\\
	\lfloor y_0 \rfloor, & \text{if} \
	p(\lfloor y_0 \rfloor-1|\boldsymbol{\theta})<
	p(\lfloor y_0 \rfloor|\boldsymbol{\theta})>
	p(\lfloor y_0 \rfloor+1|\boldsymbol{\theta}),
	\\
	\lfloor y_0 \rfloor-1, & \text{if} \
	p(\lfloor y_0 \rfloor-1|\boldsymbol{\theta})>
	p(\lfloor y_0 \rfloor|\boldsymbol{\theta})\geqslant
	p(\lfloor y_0 \rfloor+1|\boldsymbol{\theta}),
	\\
	\lfloor y_0 \rfloor-1 \ \text{and} \ \lfloor y_0 \rfloor, & \text{if} \
	p(\lfloor y_0 \rfloor-1|\boldsymbol{\theta})=
	p(\lfloor y_0 \rfloor|\boldsymbol{\theta})>
	p(\lfloor y_0 \rfloor+1|\boldsymbol{\theta}),
	\\
	\lfloor y_0 \rfloor \ \text{and} \ \lfloor y_0 \rfloor+1, & \text{if} \
	p(\lfloor y_0 \rfloor-1|\boldsymbol{\theta})<
	p(\lfloor y_0 \rfloor|\boldsymbol{\theta})=
	p(\lfloor y_0 \rfloor+1|\boldsymbol{\theta}),
	\\
	\lfloor y_0 \rfloor,  & \text{if} \
	p(\lfloor y_0 \rfloor-1|\boldsymbol{\theta})=
	p(\lfloor y_0 \rfloor+1|\boldsymbol{\theta})<
	p(\lfloor y_0 \rfloor|\boldsymbol{\theta}).
	\end{cases}
	\end{align*}
	Notice that, any other possible relation between $p(\lfloor y_0 \rfloor-1|\boldsymbol{\theta})$,
	$p(\lfloor y_0 \rfloor|\boldsymbol{\theta})$
	and  
	$p(\lfloor y_0 \rfloor+1|\boldsymbol{\theta})$ contradicts the fact that $y_0$ is a mode of the continuous log-symmetric distribution.
	As mentioned above,  using \eqref{ineqs-2}, the second mode of the discrete log-symmetric distribution is obtained in an analogous way. So we completed the proof.
\end{proof}

As an immediate consequence of Theorem \ref{teo-shapes}, the following result follows.
\begin{corollary}
	Let $g$ be a density generating kernel so that the corresponding continuous log-symmetric distribution, of a  random variable $Y$, is unimodal. Then the discrete log-symmetric distribution of $X=\lfloor Y\rfloor$ is also  unimodal.
\end{corollary}


\section{Maximum likelihood estimation}\label{sec:03}

\subsection{Uncensored data}

Let $(X_{1}, \dots, X_{n})$ be a random sample of size $n$ from a random variable $X$ with
PMF given by \eqref{relation} and $\boldsymbol{x}=(x_{1}, \dots, x_{n})$ their observations (data).
Then,
the log-likelihood function for a parameter vector $\bm{\theta}=(\lambda,\phi)$ 
is
given
by
\begin{align}
\ell(\bm\theta)
=
\ell(\bm\theta|\boldsymbol{x})
&=
\sum_{i=1}^n \log
p(x_i|\boldsymbol{\theta})
=
{\sum_{i=1}^n} 
\log\left\{
G\big[a_{\boldsymbol{\theta}}(x_i+1)\big]
-
G\big[a_{\boldsymbol{\theta}}(x_i)\big]
\right\}.
\label{logvero}
\end{align}

The roots of the system formed by the partial derivatives of the log-likelihood function $\ell(\bm\theta)$ with respect to $\lambda$ and $\phi$ are the estimates of these parameters, respectively. Thus, we must solve the following system of equations:
\begin{align*}
{\partial \ell(\bm\theta)\over \partial \theta}
&=
(Z_g)^{-1}
\sum _{i=1}^n 
\sum_{j=0}^{1}
(-1)^{j+1} \,
{\partial a_{\boldsymbol{\theta}}(x_{i}+j)\over \partial \theta} \,
{g\big[a_{\boldsymbol{\theta}}^2(x_i+j)\big] \over  p(x_i|\boldsymbol{\theta})}
=0,
\quad 
\theta\in\{\lambda,\phi\},
\end{align*}
where 
$Z_g=\int_{-\infty}^{\infty}g(w^2)dw$ and,
\begin{align}\label{pri-der-a}
{\partial a_{\boldsymbol{\theta}}(x_i)\over \partial\lambda}
=
-\big(\lambda\phi^{1/2}\big)^{-1},
\quad  
{\partial a_{\boldsymbol{\theta}}(x_i)\over \partial \phi}
=
\log\Big({x_i\over\lambda}\Big)^{1/(2\phi^{3/2})}.
\end{align}
Note that they must be solved by an iterative procedure for non-linear optimization, such as the Broyden-Fletcher-Goldfarb-Shanno (BFGS) quasi-Newton method; see \citet[][p.\,199]{mjm:00}.

Inference for $\bm \theta$ of the $\textrm{LS}_{\rm d}$ model can be based on the asymptotic distribution of the maximum likelihood estimator $\widehat{\bm \theta}$. Under classic regularity conditions, this estimatior is bivariate normal distributed with mean $\bm \theta$ and covariance matrix $\bm{\Sigma}_{^{\widehat{\bm \theta}}}$, namely, 
$$
\sqrt{n}\,(\widehat{{\bm \theta}} -{\bm \theta}) \stackrel{\rm D}{\to} \textrm{N}_{2}\big(\bm{0}, \bm{\Sigma}_{^{\widehat{\bm \theta}}} = [\mathcal{J}({\bm \theta})]^{-1}\big),
$$
as $n \to \infty$, where $\stackrel{\rm D}{\to}$ means ``convergence in distribution'', $\mathcal{I}({\bm \theta})$ is the expected Fisher information matrix, and $\mathcal{J}({\bm \theta}) = \lim_{n\to\infty}({1}/{n}) [\mathcal{I}({\bm \theta})]$. Observe that $[\widehat{\mathcal{I}}({\bm \theta})]^{-1}$ is a consistent estimator of the asymptotic covariance matrix of $\widehat{\bm \theta}$. Observe also that one may use the Hessian matrix to obtain the observed version of the expected Fisher information matrix.   

The Hessian matrix of $\ell(\bm\theta)$ is given by
\begin{align*}
\big[{\ddot{\ell}_{\theta\theta'}(\bm\theta)}\big]_{2\times 2}
=
\begin{bmatrix}
{\partial^2 \ell(\bm\theta)\over\partial\lambda^2}
&
{\partial^2 \ell(\bm\theta)\over\partial\lambda \partial\phi}
\\[0,2cm]
{\partial^2 \ell(\bm\theta)\over\partial\phi \partial\lambda}
&
{\partial^2 \ell(\bm\theta)\over\partial\phi^2}
\end{bmatrix},
\end{align*} 
where its elements, for each $\theta, \theta'\in\{\lambda,\phi\}$, are 
\begin{align}
&\ddot{\ell}_{\theta\theta'}(\bm\theta)
=
(Z_g)^{-1}
\sum _{i=1}^n 
\sum_{j=0}^{1}
(-1)^{j+1}\,
\left[
{\partial^2 a_{\boldsymbol{\theta}}({x_{i}+j})\over \partial\theta \partial\theta'}	
+
\Theta_j(x_i)
-
\Omega_j(x_i)
\right]
{g\big[a^2({x_{i}+j})\big]\over p(x_i|\boldsymbol{\theta})}. \nonumber
\end{align} 
Here we adopt the following notation:
\begin{align}
& 
\Theta_j(x_i)=
2 a_{\boldsymbol{\theta}}({x_i+j})\,
g'\big[a^2({x_i+j})\big]\,
{\partial a_{\boldsymbol{\theta}}({x_i+j})\over\partial\theta} \,
{\partial a_{\boldsymbol{\theta}}({x_i+j})\over\partial\theta'};     \label{Mdef}
\\[0,1cm]
& 
\Omega_j(x_i)=
(Z_g)^{-1}
{\partial a_{\boldsymbol{\theta}}({x_i+j})\over\partial\theta}
\sum_{k=0}^{1} (-1)^{k+1}\,
{\partial a_{\boldsymbol{\theta}}({x_i+k})\over\partial\theta'}\,
{g\big[a^2({x_i+k})\big]\over p(x_i|\boldsymbol{\theta})};
\label{Mdef-1}
\end{align}
whenever the density generating kernel $g$ be differentiable.
The above second-order partial derivatives of $a_{\boldsymbol{\theta}}(\cdot)$, with respect to the parameters, are given by
\begin{align}\label{sec-der-a}
\begin{matrix}
\displaystyle
{\partial^2 a_{\boldsymbol{\theta}}(x_i)\over\partial\lambda^2}
=
\big(\lambda^2\phi^{1/2}\big)^{-1};
&
\displaystyle
{\partial^2 a_{\boldsymbol{\theta}}(x_i)\over\partial\phi^2}
=
\log\Big({x_i \over \lambda}\Big)^{-3/(4\phi^{5/2})};
\\[0,35cm]
\displaystyle
{\partial a_{\boldsymbol{\theta}}(x_i)\over\partial\lambda\partial\phi}
=
\big(2\lambda\phi^{3/2}\big)^{-1};
&
\displaystyle
{\partial a_{\boldsymbol{\theta}}(x_i)\over\partial\phi\partial\lambda}
=
\big(2\lambda\phi^{3/2}\big)^{-1}.
\end{matrix}
\end{align}

Under certain regularity conditions, the Fisher information matrix 
\begin{align*}
\big[{\mathcal{I}_{\theta\theta'}(\bm\theta)}\big]_{2\times 2}
=
-
\begin{bmatrix}
\mathbb{E}
{\partial^2 p(X|\boldsymbol{\theta}) \over\partial\lambda^2}
&
\mathbb{E}
{\partial^2 p(X|\boldsymbol{\theta})\over\partial\lambda \partial\phi}
\\[0,2cm]
\mathbb{E}
{\partial^2 p(X|\boldsymbol{\theta})\over\partial\phi \partial\lambda}
&
\mathbb{E}
{\partial^2 p(X|\boldsymbol{\theta})\over\partial\phi^2}
\end{bmatrix},
\quad X\sim\textrm{LS}_{\rm d}(\boldsymbol{\theta},g),
\end{align*} 
has elements of the following form
\begin{align*}
{\mathcal{I}_{\theta\theta'}(\bm\theta)}
=
(Z_g)^{-1}
\sum _{x=0}^\infty 
\sum_{j=0}^{1}
(-1)^{j}
\bigg[
{\partial^2 a_{\boldsymbol{\theta}}({x+j})\over \partial\theta \partial\theta'}
+
\Theta_j(x)
-
\Omega_j(x) 
\bigg]\, 
g\big[a^2({x+j})\big],
\end{align*}
for each $\theta, \theta'\in\{\lambda,\phi\}$,
where $\Theta_j(\cdot)$ and
$\Omega_j(\cdot)$ are given in \eqref{Mdef} and \eqref{Mdef-1}, respectively, and 
whenever the above series converges absolutely.

The extra parameter $\xi$ (or parameter vector $\bm\xi$) associated with $g$ is selected by using the profile log-likelihood function. For instance, in the case of the discrete log-Student-$t$ distribution, two steps are 
require:
\begin{itemize}
	\item[i)] Let $\xi_{k}=k$ and for each $k=1,..,100$ compute the $k$-th maximum likelihood estimate of ${\bm{\theta}}_k=({\lambda}_k,{\phi}_k)^\intercal$, $\widehat{\bm{\theta}}_k=(\widehat{\lambda}_k,\widehat{\phi}_k)^\intercal$ say. Compute also the $k$-th log-likelihood function value $\ell_k(\widehat{\bm\theta}_k)$;
	\item[ii)] The final estimate of $\xi$, $\widehat{\xi}=\xi_k$ say, is the one which maximizes the log-likelihood function, that is, $\widehat{\xi}\in \{\mbox{argmax}_{\xi_k} \ell_k(\widehat{\bm\theta}_k)\}$,  and the estimate of $\bm{\theta}$ is $\widehat{\bm{\theta}}_k=(\widehat{\lambda}_k,\widehat{\phi}_k)^\intercal$.
\end{itemize}

\subsection{Censored data}

Let $Y_i\sim\textrm{LS}(\boldsymbol{\theta},g)$ be the failure time of the $i$-th 
individual and let $\delta_i$ indicate whether 
the $i$-th individual is censored or not. 
Let us define
$d_k =$ ``number of failures at time $t_k$'',
$q_k =$ ``number censored at time $t_k$'' and
$N_k = \sum_{i=k}^{\infty} (d_i + q_i)$.
Note that $N_k-d_k$ represents the number survived just before time $t_k+1$.
That is, in each given time $t_k$, there are $d_k$ failures and $N_k-d_k$ survivals.

Since the data are discrete observing $\{(Y_i,\delta_i)\}$ is equivalent to observing $\{(d_k, q_k)\}$,
the likelihood function for the random censoring is given by
\begin{align*}
L^R(\bm\theta)
=
{\displaystyle\prod_{i=1}^n} 
\big[f_{Y}(y_i|\boldsymbol{\theta})\big]^{\delta_i}
\big[1-F_{Y}(y_i|\boldsymbol{\theta})\big]^{1-\delta_i}
=
{\displaystyle\prod_{k=1}^\infty} 
\big[p(x_k|\boldsymbol{\theta})\big]^{d_k} 
\big[p(x_k|\boldsymbol{\theta})+R(x_k|\boldsymbol{\theta})\big]^{q_k}.
\end{align*}
This type of censoring has as special
case type I and II censoring.
The corresponding log-likelihood is
\begin{align}\label{log-lik-cens}
\ell^R(\bm\theta)
&=
{\sum_{k=1}^\infty} 
\Big\{
{d_k}
\log
p(x_k|\boldsymbol{\theta})
+ 
{q_k}
\log\big[p(x_k|\boldsymbol{\theta})+R(x_k|\boldsymbol{\theta})\big]
\Big\}
\\[0,15cm]
&=
{\sum_{k=1}^\infty} 
{d_k}
\log
\big\{
G\big[a_{\boldsymbol{\theta}}(x_k+1)\big]
-
G\big[a_{\boldsymbol{\theta}}(x_k)\big]
\big\}
+
{\sum_{k=1}^\infty} 
{q_k}
\log\big\{1-G\big[a_{\boldsymbol{\theta}}(x_k)\big] \big\}
\nonumber.
\end{align}
\begin{remark}
	By Proposition \ref{chac-re}, the log-likelihood \eqref{log-lik-cens} can be rewritten in terms of HR as
	\[
	\ell^R(\bm\theta)
	=
	{\sum_{k=1}^\infty} 
	(d_k+q_k)
	\bigg\{
	{d_k\over d_k+q_k} \log\big[ H(x_k|\boldsymbol{\theta})\big]
	-
	\log\big[1- H(x_k|\boldsymbol{\theta})\big]
	+
	\sum_{y=0}^{x_k-1}
	\log\big[1- H(y|\boldsymbol{\theta})\big]
	\bigg\},
	\]
	whenever the above series converges absolutely.
\end{remark}

Differentiating  in \eqref{log-lik-cens}, a straightforward computation shows that 
\begin{align*}
{\partial \ell^R(\bm\theta)\over\partial\theta}
&=
(Z_g)^{-1}
\sum _{k=1}^\infty
d_k
\sum_{j=0}^{1}
(-1)^{j+1}\,
{\partial a_{\boldsymbol{\theta}}(x_{k}+j)\over\partial\theta}\,
{g\big[a_{\boldsymbol{\theta}}^2(x_k+j)\big] \over  p(x_k|\boldsymbol{\theta})}
\\[0,1cm]
&
-
(Z_g)^{-1}
\sum _{k=1}^\infty
q_k\, 
{\partial a_{\boldsymbol{\theta}}(x_{k})\over\partial\theta} \,
{g\big[a_{\boldsymbol{\theta}}^2(x_k)\big] \over  
	p(x_k|\boldsymbol{\theta})+R(x_k|\boldsymbol{\theta}) },
\quad \theta\in\{\lambda,\phi\},
\end{align*}
where $Z_g=\int^{\infty}_{-\infty} g(w^2)  \,\textrm{d}w$.
The mixed partial derivatives of $\ell^R(\bm\theta)$ are given by
\begin{align}
{\partial^2\ell^R(\bm\theta)\over \partial\theta\partial\theta' }
&=
(Z_g)^{-1}
\sum _{k=1}^\infty
d_k
\sum_{j=0}^{1}
(-1)^{j+1}\,
\bigg[
{\partial^2 a_{\boldsymbol{\theta}}({x_{k}+j})\over \partial\theta \partial\theta'}
+
\Theta_j(x_k)
-
\Omega_j(x_k)
\bigg]\, 
{g\big[a^2_{\boldsymbol{\theta}}({x_{k}+j})\big]\over p(x_k|\boldsymbol{\theta})}
\nonumber
\\[0,1cm]
&
-
(Z_g)^{-1}
\sum _{k=1}^\infty
q_k\, 
\bigg[
{\partial^2 a_{\boldsymbol{\theta}}({x_{k}})\over \partial\theta \partial\theta'}
+
\Theta_0(x_k)
+
\hat{\Omega}(x_k)
\bigg]\,
{g\big[a_{\boldsymbol{\theta}}^2(x_k)\big] \over  
	p(x_k|\boldsymbol{\theta})+R(x_k|\boldsymbol{\theta}) },
\quad \theta, \theta'\in\{\lambda,\phi\},
\nonumber
\end{align}
where
\[
\hat{\Omega}(x_k)=
(Z_g)^{-1}
{g\big[a^2({x_k})\big]\over p(x_k|\boldsymbol{\theta})+R(x_k|\boldsymbol{\theta}) } \,
{\partial a_{\boldsymbol{\theta}}({x_k})\over\partial\theta}
{\partial a_{\boldsymbol{\theta}}({x_k})\over\partial\theta'},
\]
and $\Theta_j(\cdot), \Omega_j(\cdot)$ are as in \eqref{Mdef} and \eqref{Mdef-1}, respectively. The first and second derivatives of function $a_{\boldsymbol{\theta}}({\cdot})$ are given in \eqref{pri-der-a} and \eqref{sec-der-a}, respectively.

\section{Monte Carlo simulation study}\label{sec:04}

A Monte Carlo simulation study was carried out to evaluate the performance of the maximum likelihood estimators for the $\textrm{LS}_{\rm d}$ models, particularly the log-normal, log-Student-$t$, 
log-contaminated-normal, log-power-exponential, extended Birnbaum-Saunders and 
extended Birnbaum-Saunders-$t$ cases. Note that when $\phi=4$ (fixed) the Birnbaum-Saunders and Birnbaum-Saunders-$t$ are obtained. All numerical evaluations were done in the \texttt{R} software; see \cite{R:15}.

The simulation scenario considers: sample size $n \in \{40,120,400\}$, values of true 
parameters $\phi \in \{1,4,8\}$, $\lambda \in \{2.00\}$, censoring proportions $\{0\%,10\%,30\%\}$, and $1,000$ Monte Carlo replications for each sample size. The values of the true extra parameters are presented in the caption of each table. 

The maximum likelihood estimation results are presented in Tables~\ref{t1}--\ref{t6}.  The following sample statistics for the maximum likelihood estimates are reported: empirical mean, 
bias, and mean squared error (MSE). A look at the results in Tables~\ref{t1}--\ref{t6} allows us to conclude that, as the sample size increases, the bias and MSE of all the estimators decrease, indicating that they are asymptotically unbiased, as expected. Moreover, as the censoring proportion increases, the performances of the estimators of $\phi$ and $\lambda$, deteriorate. Generally, all of these results show the 
good performance of the proposed model.

\begin{table}[H]
	\caption{Empirical values of mean, bias and MSE from simulated discrete log-normal data for the indicated maximum likelihood estimators.}
	\label{t1}
	\scalefont{0.82}
	\centering
	\begin{tabular}{lccccccccccccc}
		\hline
		\multirow{2}{*}{\thead{n}}&\multirow{2}{*}{Cen.}&&\multicolumn{3}{c}{$\phi=1$}&&\multicolumn{3}{c}{$\phi=4$}&&\multicolumn{3}{c}{$\phi=8$}\\
		\cline{4-6}\cline{8-10}\cline{12-14}
		&&&Mean&Bias&MSE&&Mean&Bias&MSE&&Mean&Bias&MSE\\
		\cline{4-14}
		\multirow{2}{*}{40}&\multirow{6}{*}{0\%}&$\hat{\phi}$&1.0119&0.0119&0.0851&&4.0752&0.0752&1.7247&&8.1388&0.1388&7.5728\\
		&&$\hat{\lambda}$&2.0262&0.0262&0.1166&&2.1126&0.1126&0.5629&&2.2586&0.2586&1.4043\\
		\multirow{2}{*}{120}&&$\hat{\phi}$&1.0059&0.0059&0.0260&&4.0413&0.0413&0.5237&&8.0768&0.0768&2.1683\\
		&&$\hat{\lambda}$&2.0086&0.0086&0.0394&&2.0317&0.0317&0.1759&&2.0750&0.0750&0.3897\\
		\multirow{2}{*}{400}&&$\hat{\phi}$&1.0041&0.0041&0.0077&&4.0184&0.0184&0.1522&&8.0404&0.0404&0.6641\\
		&&$\hat{\lambda}$&2.0021&0.0021&0.0110&&2.0086&0.0086&0.0464&&2.0198&0.0198&0.1020\\
		\hline
		\multirow{2}{*}{40}&\multirow{6}{*}{10\%}&$\hat{\phi}$&1.0189&0.0189&0.0998&&4.1507&0.1507&2.2755&&8.3002&0.3002&10.0324\\
		&&$\hat{\lambda}$&2.0301&0.0301&0.1198&&2.1154&0.1154&0.5850&&2.2528&0.2528&1.4323\\
		\multirow{2}{*}{120}&&$\hat{\phi}$&1.0114&0.0114&0.0290&&4.0519&0.0519&0.5978&&8.1162&0.1162&2.5822\\
		&&$\hat{\lambda}$&2.0099&0.0099&0.0396&&2.0416&0.0416&0.1795&&2.0842&0.0842&0.4037\\
		\multirow{2}{*}{400}&&$\hat{\phi}$&1.0015&0.0015&0.0090&&4.0262&0.0262&0.1817&&8.0675&0.0675&0.8152\\
		&&$\hat{\lambda}$&2.0012&0.0012&0.0109&&2.0014&0.0014&0.0466&&2.0073&0.0073&0.0984\\
		\hline
		\multirow{2}{*}{40}&\multirow{6}{*}{30\%}&$\hat{\phi}$&1.0562&0.0562&0.1839&&4.4113&0.4113&5.7506&&8.9965&0.9965&28.9825\\
		&&$\hat{\lambda}$&2.0435&0.0435&0.1407&&2.1396&0.1396&0.6858&&2.2925&0.2925&1.7228\\
		\multirow{2}{*}{120}&&$\hat{\phi}$&1.0218&0.0218&0.0458&&4.1116&0.1116&1.1065&&8.2539&0.2539&5.0762\\
		&&$\hat{\lambda}$&2.0124&0.0124&0.0438&&2.0450&0.0450&0.1915&&2.0868&0.0868&0.4270\\
		\multirow{2}{*}{400}&&$\hat{\phi}$&1.0093&0.0093&0.0133&&4.0840&0.0840&0.3158&&8.1961&0.1961&1.5102\\
		&&$\hat{\lambda}$&2.0046&0.0046&0.0119&&2.0076&0.0076&0.0501&&2.0131&0.0131&0.1043\\
		\hline
	\end{tabular}
\end{table}

\begin{table}[H]
	\caption{Empirical values of mean, bias and MSE from simulated discrete log-Student-$t$ data for the indicated maximum likelihood estimators with $\xi=4$.}
	\label{t2}
	\scalefont{0.82}
	\centering
	\begin{tabular}{lccccccccccccc}
		\hline
		\multirow{2}{*}{\thead{$n$}}&\multirow{2}{*}{Cen.}&&\multicolumn{3}{c}{$\phi=1$}&&\multicolumn{3}{c}{$\phi=4$}&&\multicolumn{3}{c}{$\phi=8$}\\
		\cline{4-6}\cline{8-10}\cline{12-14}
		&&&Mean&Bias&MSE&&Mean&Bias&MSE&&Mean&Bias&MSE\\
		\cline{4-14}
		\multirow{2}{*}{40}&\multirow{6}{*}{0\%}&$\hat{\phi}$&1.0250&0.0250&0.1313&&4.1432&0.1432&2.6039&&8.3035&0.3035&14.7992\\
		&&$\hat{\lambda}$&2.0351&0.0351&0.1589&&2.1416&0.1416&0.7463&&2.3098&0.3098&1.9332\\
		\multirow{2}{*}{120}&&$\hat{\phi}$&1.0093&0.0093&0.0408&&4.0316&0.0316&0.8793&&7.9743&-0.0257&4.0502\\
		&&$\hat{\lambda}$&2.0179&0.0179&0.0497&&2.0657&0.0657&0.2195&&2.1694&0.1694&0.7087\\
		\multirow{2}{*}{400}&&$\hat{\phi}$&0.9974&-0.0026&0.0123&&3.9778&-0.0222&0.4006&&7.9731&-0.0269&2.8371\\
		&&$\hat{\lambda}$&1.9986&-0.0014&0.0137&&2.0148&0.0148&0.0684&&2.0657&0.0657&0.2455\\
		\hline
		\multirow{2}{*}{40}&\multirow{6}{*}{10\%}&$\hat{\phi}$&1.0332&0.0332&0.1499&&4.2117&0.2117&3.2204&&8.5131&0.5131&14.0700\\
		&&$\hat{\lambda}$&2.0318&0.0318&0.1692&&2.1338&0.1338&0.8105&&2.2965&0.2965&2.1397\\
		\multirow{2}{*}{120}&&$\hat{\phi}$&1.0149&0.0149&0.0466&&4.0564&0.0564&0.9162&&8.1134&0.1134&4.0340\\
		&&$\hat{\lambda}$&2.0068&0.0068&0.0510&&2.0434&0.0434&0.2086&&2.0916&0.0916&0.4621\\
		\multirow{2}{*}{400}&&$\hat{\phi}$&1.0026&0.0026&0.0126&&4.0285&0.0285&0.2729&&8.0335&0.0335&1.1661\\
		&&$\hat{\lambda}$&1.9987&-0.0013&0.0146&&2.0000&0.0000&0.0598&&2.0137&0.0137&0.1208\\
		\hline
		\multirow{2}{*}{40}&\multirow{6}{*}{30\%}&$\hat{\phi}$&1.0772&0.0772&0.2681&&4.5456&0.5456&8.5155&&9.3122&1.3122&41.0324\\
		&&$\hat{\lambda}$&2.0430&0.0430&0.2007&&2.1605&0.1605&1.0059&&2.3570&0.3570&3.1707\\
		\multirow{2}{*}{120}&&$\hat{\phi}$&1.0148&0.0148&0.0666&&4.0801&0.0801&1.4871&&8.1426&0.1426&6.8379\\
		&&$\hat{\lambda}$&2.0032&0.0032&0.0534&&2.0429&0.0429&0.2210&&2.0866&0.0866&0.4900\\
		\multirow{2}{*}{400}&&$\hat{\phi}$&1.0129&0.0129&0.0168&&4.0890&0.0890&0.4097&&8.1679&0.1679&1.8873\\
		&&$\hat{\lambda}$&2.0021&0.0021&0.0160&&2.0040&0.0040&0.0645&&2.0179&0.0179&0.1297\\
		\hline
	\end{tabular}
\end{table}

\begin{table}[H]
	\caption{Empirical values of mean, bias and MSE from simulated discrete log-contaminated-normal data for the indicated maximum likelihood estimators with $\bm{\xi}=(0.5,0.5)^\intercal$.}
	\label{t3}
	\scalefont{0.82}
	\centering
	\begin{tabular}{lccccccccccccc}
		\hline
		\multirow{2}{*}{\thead{$n$}}&\multirow{2}{*}{Cen.}&&\multicolumn{3}{c}{$\phi=1$}&&\multicolumn{3}{c}{$\phi=4$}&&\multicolumn{3}{c}{$\phi=8$}\\
		\cline{4-6}\cline{8-10}\cline{12-14}
		&&&Mean&Bias&MSE&&Mean&Bias&MSE&&Mean&Bias&MSE\\
		\cline{4-14}
		\multirow{2}{*}{40}&\multirow{6}{*}{0\%}&$\hat{\phi}$&1.0057&0.0057&0.0947&&4.0459&0.0459&1.8954&&8.0944&0.0944&7.9832\\
		&&$\hat{\lambda}$&2.0437&0.0437&0.1706&&2.1750&0.1750&0.8474&&2.3687&0.3687&2.1616\\
		\multirow{2}{*}{120}&&$\hat{\phi}$&1.0024&0.0024&0.0300&&4.0207&0.0207&0.6042&&8.0165&0.0165&2.5295\\
		&&$\hat{\lambda}$&2.0161&0.0161&0.0545&&2.0589&0.0589&0.2503&&2.1297&0.1297&0.5661\\
		\multirow{2}{*}{400}&&$\hat{\phi}$&1.0020&0.0020&0.0085&&4.0061&0.0061&0.1779&&7.9912&-0.0088&0.8225\\
		&&$\hat{\lambda}$&2.0024&0.0024&0.0165&&2.0161&0.0161&0.0745&&2.0410&0.0410&0.1751\\
		\hline
		\multirow{2}{*}{40}&\multirow{6}{*}{10\%}&$\hat{\phi}$&1.0336&0.0336&0.1484&&4.1971&0.1971&3.1525&&8.4153&0.4153&12.8221\\
		&&$\hat{\lambda}$&2.0562&0.0562&0.1819&&2.1952&0.1952&0.9257&&2.3944&0.3944&2.3858\\
		\multirow{2}{*}{120}&&$\hat{\phi}$&1.0154&0.0154&0.0396&&4.0676&0.0676&0.7806&&8.1552&0.1552&3.3793\\
		&&$\hat{\lambda}$&2.0182&0.0182&0.0588&&2.0688&0.0688&0.2649&&2.1338&0.1338&0.5975\\
		\multirow{2}{*}{400}&&$\hat{\phi}$&1.0067&0.0067&0.0103&&4.0278&0.0278&0.2113&&8.0581&0.0581&0.9135\\
		&&$\hat{\lambda}$&1.9988&-0.0012&0.0163&&2.0077&0.0077&0.0735&&2.0222&0.0222&0.1586\\
		\hline
		\multirow{2}{*}{40}&\multirow{6}{*}{30\%}&$\hat{\phi}$&1.0644&0.0644&0.2448&&4.4774&0.4774&7.4371&&9.1280&1.1280&33.9336\\
		&&$\hat{\lambda}$&2.0635&0.0635&0.2088&&2.2106&0.2106&1.0845&&2.4422&0.4422&3.1094\\
		\multirow{2}{*}{120}&&$\hat{\phi}$&1.0310&0.0310&0.0591&&4.1681&0.1681&1.3522&&8.4147&0.4147&6.3718\\
		&&$\hat{\lambda}$&2.0237&0.0237&0.0635&&2.0767&0.0767&0.2895&&2.1436&0.1436&0.6536\\
		\multirow{2}{*}{400}&&$\hat{\phi}$&1.0126&0.0126&0.0161&&4.0609&0.0609&0.3648&&8.1301&0.1301&1.6252\\
		&&$\hat{\lambda}$&2.0005&0.0005&0.0176&&2.0102&0.0102&0.0771&&2.0237&0.0237&0.1628\\
		\hline
	\end{tabular}
\end{table}

\begin{table}[H]
	\caption{Empirical values of mean, bias and MSE from simulated discrete log-power-exponential data for the indicated maximum likelihood estimators with $\xi=-0.5$.}
	\label{t4}
	\scalefont{0.82}
	\centering
	\begin{tabular}{lccccccccccccc}
		\hline
		\multirow{2}{*}{\thead{$n$}}&\multirow{2}{*}{Cen.}&&\multicolumn{3}{c}{$\phi=1$}&&\multicolumn{3}{c}{$\phi=4$}&&\multicolumn{3}{c}{$\phi=8$}\\
		\cline{4-6}\cline{8-10}\cline{12-14}
		&&&Mean&Bias&MSE&&Mean&Bias&MSE&&Mean&Bias&MSE\\
		\cline{4-14}
		\multirow{2}{*}{40}&\multirow{6}{*}{0\%}&$\hat{\phi}$&0.9783&-0.0217&0.0470&&3.9497&-0.0503&1.1520&&7.9186&-0.0814&5.6648\\
		&&$\hat{\lambda}$&2.0072&0.0072&0.0481&&2.0493&0.0493&0.2856&&2.1332&0.1332&0.6806\\
		\multirow{2}{*}{120}&&$\hat{\phi}$&0.9970&-0.0030&0.0134&&4.0014&0.0014&0.3399&&7.9966&-0.0034&1.5695\\
		&&$\hat{\lambda}$&2.0032&0.0032&0.0153&&2.0185&0.0185&0.0863&&2.0482&0.0482&0.1996\\
		\multirow{2}{*}{400}&&$\hat{\phi}$&0.9969&-0.0031&0.0042&&3.9906&-0.0094&0.1019&&7.9864&-0.0136&0.4742\\
		&&$\hat{\lambda}$&1.9998&-0.0002&0.0047&&2.0036&0.0036&0.0282&&2.0098&0.0098&0.0640\\
		\hline
		\multirow{2}{*}{40}&\multirow{6}{*}{10\%}&$\hat{\phi}$&0.9941&-0.0059&0.0612&&4.0123&0.0123&1.5967&&8.0169&0.0169&8.0172\\
		&&$\hat{\lambda}$&2.0046&0.0046&0.0510&&2.0472&0.0472&0.2900&&2.1398&0.1398&0.6989\\
		\multirow{2}{*}{120}&&$\hat{\phi}$&0.9999&-0,0001&0.0167&&4.0151&0.0151&0.4585&&8.0260&0.0260&2.0633\\
		&&$\hat{\lambda}$&2.0010&0.0010&0.0165&&2.0099&0.0099&0.0866&&2.0349&0.0349&0.1968\\
		\multirow{2}{*}{400}&&$\hat{\phi}$&1.0020&0.0020&0.0053&&4.0178&0.0178&0.1283&&8.0441&0.0441&0.6083\\
		&&$\hat{\lambda}$&2.0010&0.0010&0.0050&&2.0044&0.0044&0.0298&&2.0084&0.0084&0.0661\\
		\hline
		\multirow{2}{*}{40}&\multirow{6}{*}{30\%}&$\hat{\phi}$&1.0093&0.0093&0.0890&&4.1602&0.1602&3.3402&&8.2982&0.2982&15.8037\\
		&&$\hat{\lambda}$&2.0111&0.0111&0.0646&&2.0512&0.0512&0.3185&&2.1408&0.1408&0.7509\\
		\multirow{2}{*}{120}&&$\hat{\phi}$&1.0040&0.0040&0.0250&&4.0485&0.0485&0.7736&&8.1109&0.1109&3.7337\\
		&&$\hat{\lambda}$&2.0019&0.0019&0.0191&&2.0104&0.0104&0.0927&&2.0350&0.0350&0.2090\\
		\multirow{2}{*}{400}&&$\hat{\phi}$&1.0026&0.0026&0.0079&&4.0215&0.0215&0.1934&&8.0665&0.0665&0.9711\\
		&&$\hat{\lambda}$&2.0010&0.0010&0.0056&&2.0045&0.0045&0.0318&&2.0093&0.0093&0.0694\\
		\hline
	\end{tabular}
\end{table}

\begin{table}[H]
	\caption{Empirical values of mean, bias and MSE from simulated discrete extended Birnbaum-Saunders data for the indicated maximum likelihood estimators with $\zeta=0.5$.}
	\label{t5}
	\scalefont{0.82}
	\centering
	\begin{tabular}{lccccccccccccc}
		\hline
		\multirow{2}{*}{\thead{$n$}}&\multirow{2}{*}{Cen.}&&\multicolumn{3}{c}{$\phi=1$}&&\multicolumn{3}{c}{$\phi=4$}&&\multicolumn{3}{c}{$\phi=8$}\\
		\cline{4-6}\cline{8-10}\cline{12-14}
		&&&Mean&Bias&MSE&&Mean&Bias&MSE&&Mean&Bias&MSE\\
		\cline{4-14}
		\multirow{2}{*}{40}&\multirow{6}{*}{0\%}&$\hat{\phi}$&0.8741&-0.1259&0.1721&&3.9719&-0.0281&0.9531&&7.9871&-0.0129&4.3427\\
		&&$\hat{\lambda}$&2.0126&0.0126&0.0075&&2.0117&0.0117&0.0275&&2.0152&0.0152&0.0542\\
		\multirow{2}{*}{120}&&$\hat{\phi}$&0.9747&-0.0253&0.0353&&4.0090&0.0090&0.3178&&8.0289&0.0289&1.3258\\
		&&$\hat{\lambda}$&2.0049&0.0049&0.0029&&2.0032&0.0032&0.0090&&2.0037&0.0037&0.0185\\
		\multirow{2}{*}{400}&&$\hat{\phi}$&0.9913&-0.0087&0.0117&&4.0042&0.0042&0.0978&&8.0266&0.0266&0.4048\\
		&&$\hat{\lambda}$&2.0020&0.0020&0.0009&&2.0011&0.0011&0.0026&&2.0004&0.0004&0.0051\\
		\hline
		\multirow{2}{*}{40}&\multirow{6}{*}{10\%}&$\hat{\phi}$&0.8825&-0.1175&0.1945&&4.0107&0.0107&1.2217&&8.1117&0.1117&5.3717\\
		&&$\hat{\lambda}$&2.0128&0.0128&0.0076&&2.0120&0.0120&0.0289&&2.0172&0.0172&0.0565\\
		\multirow{2}{*}{120}&&$\hat{\phi}$&0.9699&-0.0301&0.0460&&4.0048&0.0048&0.3449&&8.0872&0.0872&1.6128\\
		&&$\hat{\lambda}$&2.0047&0.0047&0.0028&&2.0043&0.0043&0.0096&&2.0031&0.0031&0.0187\\
		\multirow{2}{*}{400}&&$\hat{\phi}$&0.9907&-0.0093&0.0134&&4.0000&0.0000&0.1167&&8.0092&0.0092&0.4932\\
		&&$\hat{\lambda}$&2.0005&0.0005&0.0009&&1.9997&-0.0003&0.0026&&2.0005&0.0005&0.0050\\
		\hline
		\multirow{2}{*}{40}&\multirow{6}{*}{30\%}&$\hat{\phi}$&0.7674&-0.2326&0.3577&&4.0841&0.0841&2.0748&&8.3517&0.3517&8.9625\\
		&&$\hat{\lambda}$&2.0152&0.0152&0.0067&&2.0178&0.0178&0.0343&&2.0280&0.0280&0.0692\\
		\multirow{2}{*}{120}&&$\hat{\phi}$&0.9263&-0.0737&0.0967&&4.0323&0.0323&0.5129&&8.1723&0.1723&2.5286\\
		&&$\hat{\lambda}$&2.0052&0.0052&0.0028&&2.0057&0.0057&0.0106&&2.0055&0.0055&0.0210\\
		\multirow{2}{*}{400}&&$\hat{\phi}$&0.9852&-0.0148&0.0212&&4.0188&0.0188&0.1660&&8.0638&0.0638&0.7290\\
		&&$\hat{\lambda}$&2.0009&0.0009&0.0009&&2.0011&0.0011&0.0028&&2.0028&0.0028&0.0056\\
		\hline
	\end{tabular}
\end{table}

\begin{table}[H]
	\caption{Empirical values of mean, bias and MSE from simulated discrete extended Birnbaum-Saunders-$t$ data for the indicated maximum likelihood estimators with $\bm{\xi}=(0.5,4)^\intercal$.}
	\label{t6}
	\scalefont{0.82}
	\centering
	\begin{tabular}{lccccccccccccc}
		\hline
		\multirow{2}{*}{\thead{$n$}}&\multirow{2}{*}{Cen.}&&\multicolumn{3}{c}{$\phi=1$}&&\multicolumn{3}{c}{$\phi=4$}&&\multicolumn{3}{c}{$\phi=8$}\\
		\cline{4-6}\cline{8-10}\cline{12-14}
		&&&Mean&Bias&MSE&&Mean&Bias&MSE&&Mean&Bias&MSE\\
		\cline{4-14}
		\multirow{2}{*}{40}&\multirow{6}{*}{0\%}&$\hat{\phi}$&0.9871&-0.0129&0.1511&&3.9863&-0.0137&1.4404&&8.0258&0.0258&6.0994\\
		&&$\hat{\lambda}$&2.0028&0.0028&0.0114&&2.0100&0.0100&0.0387&&2.0167&0.0167&0.0756\\
		\multirow{2}{*}{120}&&$\hat{\phi}$&0.9998&-0.0002&0.0457&&4.0319&0.0319&0.4972&&8.0689&0.0689&2.1333\\
		&&$\hat{\lambda}$&2.0026&0.0026&0.0037&&2.0061&0.0061&0.0128&&2.0086&0.0086&0.0248\\
		\multirow{2}{*}{400}&&$\hat{\phi}$&0.9984&-0.0016&0.0141&&3.9996&-0.0004&0.1569&&8.0282&0.0282&0.6348\\
		&&$\hat{\lambda}$&2.0018&0.0018&0.0010&&2.0031&0.0031&0.0038&&2.0035&0.0035&0.0074\\
		\hline
		\multirow{2}{*}{40}&\multirow{6}{*}{10\%}&$\hat{\phi}$&1.0002&0.0002&0.1663&&4.1093&0.1093&1.8263&&8.2395&0.2395&7.9333\\
		&&$\hat{\lambda}$&2.0009&0.0009&0.0109&&2.0029&0.0029&0.0387&&2.0095&0.0095&0.0760\\
		\multirow{2}{*}{120}&&$\hat{\phi}$&1.0055&0.0055&0.0512&&4.0492&0.0492&0.5522&&8.1287&0.1287&2.3725\\
		&&$\hat{\lambda}$&2.0027&0.0027&0.0037&&2.0044&0.0044&0.0129&&2.0080&0.0080&0.0245\\
		\multirow{2}{*}{400}&&$\hat{\phi}$&1.0003&0.0003&0.0157&&3.9944&-0.0056&0.1649&&8.0085&0.0085&0.6848\\
		&&$\hat{\lambda}$&2.0011&0.0011&0.0012&&2.0020&0.0020&0.0042&&2.0028&0.0028&0.0078\\
		\hline
		\multirow{2}{*}{40}&\multirow{6}{*}{30\%}&$\hat{\phi}$&0.9829&-0.0171&0.2728&&4.1695&0.1695&2.7474&&8.4158&0.4158&12.0941\\
		&&$\hat{\lambda}$&2.0027&0.0027&0.0114&&2.0055&0.0055&0.0438&&2.0151&0.0151&0.0882\\
		\multirow{2}{*}{120}&&$\hat{\phi}$&1.0115&0.0115&0.0708&&4.1024&0.1024&0.7985&&8.2566&0.2566&3.5320\\
		&&$\hat{\lambda}$&2.0029&0.0029&0.0038&&2.0065&0.0065&0.0140&&2.0119&0.0119&0.0273\\
		\multirow{2}{*}{400}&&$\hat{\phi}$&1.0009&0.0009&0.0238&&4.0094&0.0094&0.2337&&8.0494&0.0494&0.9990\\
		&&$\hat{\lambda}$&2.0014&0.0014&0.0012&&2.0037&0.0037&0.0047&&2.0049&0.0049&0.0087\\
		\hline
	\end{tabular}
\end{table}

\section{Illustrative examples}\label{sec:05}

The $\textrm{LS}_{\rm d}$ models are now used to analyze two real-world data sets. It is considered the following discrete $\textrm{LS}_{\rm d}$ models: 
log-normal (LN), 
log-Student-$t$ (L$t$),
log-contamined-normal (LCN),
log-power-exponential (LPE), 
Birnbaum-Saunders (BS),
extended Birnbaum-Saunders (EBS),
Birnbaum-Saunders-$t$ (BS$t$), and
extended Birnbaum-Saunders-$t$ (EBS$t$).

\begin{Example}
	
	The first data set corresponds to the number of times that a DEC-20 computer broke down in each of 128 consecutive weeks of operation. This computer has operated at the Open University during the 1980s; see Table \ref{t:data2} and \cite{trenkler:95}. Descriptive statistics for the computer breaks data set are the following: $128$(sample size), $0$(minimum), $22$(maximum), $3$(median), $4.016$(mean), $3.808$(standard deviation), $94.839$(coefficient of variation), $1.732$(coefficient of skewness) and $3.995$(coefficient of kurtosis). 
	From these results, we observe the positive skewness and a high degree of kurtosis. Figure~\ref{fig:ex2}(left) shows the histogram for the computer breaks data, from where it is confirmed the positive skewness. Moreover, Figure~\ref{fig:ex2}(right) shows the usual and adjusted boxplots, and we note that some potential outliers are not in fact outliers when the adjusted boxplot is observed.

	\begin{table}[ht!]
		\caption{Computer breaks data.}
		\label{t:data2}
		\scalefont{0.82}
		\centering
		\begin{tabular}{lccccccccccccccccc}
			\hline
			x&0 & 1  &2  &3  &4  &5  &6  &7  &8  &9 &10 &11 &12 &13 &16 &17 &22\\ 
			frequency&15 &19 &23 &14 &15 &10  &8  &4  &6  &2  &3  &3  &2  &1  &1  &1  &1\\
			\hline
		\end{tabular}
	\end{table}

	\begin{figure}[!ht]
		\centering
		\psfrag{0.00}[c][c]{\scriptsize{0.00}}
		\psfrag{0.05}[c][c]{\scriptsize{0.05}}
		\psfrag{0.10}[c][c]{\scriptsize{0.10}}
		\psfrag{0.15}[c][c]{\scriptsize{0.15}}
		\psfrag{0.20}[c][c]{\scriptsize{0.20}}
		\psfrag{0.25}[c][c]{\scriptsize{0.25}}
		\psfrag{0.30}[c][c]{\scriptsize{0.30}}
		\psfrag{0.0}[c][c]{\scriptsize{0.0}}
		\psfrag{0.1}[c][c]{\scriptsize{0.1}}
		\psfrag{0.2}[c][c]{\scriptsize{0.2}}
		\psfrag{0.3}[c][c]{\scriptsize{0.3}}
		\psfrag{0.4}[c][c]{\scriptsize{0.4}}
		\psfrag{0.5}[c][c]{\scriptsize{0.5}}
		\psfrag{0.6}[c][c]{\scriptsize{0.6}}
		\psfrag{0.7}[c][c]{\scriptsize{0.7}}
		\psfrag{0.8}[c][c]{\scriptsize{0.8}}
		\psfrag{1.0}[c][c]{\scriptsize{1.0}}
		\psfrag{0}[c][c]{\scriptsize{0}} \psfrag{2}[c][c]{\scriptsize{2}}
		\psfrag{4}[c][c]{\scriptsize{4}} \psfrag{5}[c][c]{\scriptsize{5}}
		\psfrag{6}[c][c]{\scriptsize{6}} \psfrag{8}[c][c]{\scriptsize{8}}
		\psfrag{0}[c][c]{\scriptsize{0}}
		\psfrag{10}[c][c]{\scriptsize{10}}
		\psfrag{15}[c][c]{\scriptsize{15}}
		\psfrag{20}[c][c]{\scriptsize{20}}
		\psfrag{30}[c][c]{\scriptsize{30}}
		\psfrag{40}[c][c]{\scriptsize{40}}
		\psfrag{50}[c][c]{\scriptsize{50}}
		\psfrag{100}[c][c]{\scriptsize{100}}
		\psfrag{150}[c][c]{\scriptsize{150}}
		\psfrag{dp}[c][c]{\scriptsize{$x$}}
		\psfrag{c}[c][c]{\scriptsize{$x$}}
		\psfrag{dc}[c][c]{\scriptsize{Frequency}}
		\psfrag{aa}[c][c]{\scriptsize{usual boxplot}}
		\psfrag{ad}[c][c]{\scriptsize{adjusted boxplot}}
		{\includegraphics[height=7.2cm,width=7.2cm]{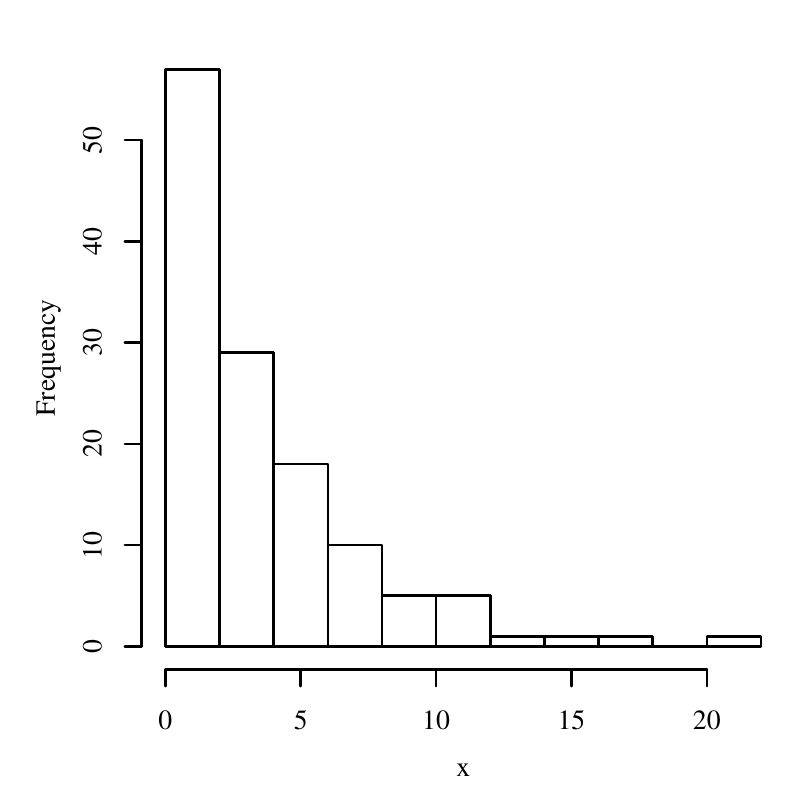}}\hspace{-0.25cm}
		{\includegraphics[height=7.2cm,width=7.2cm]{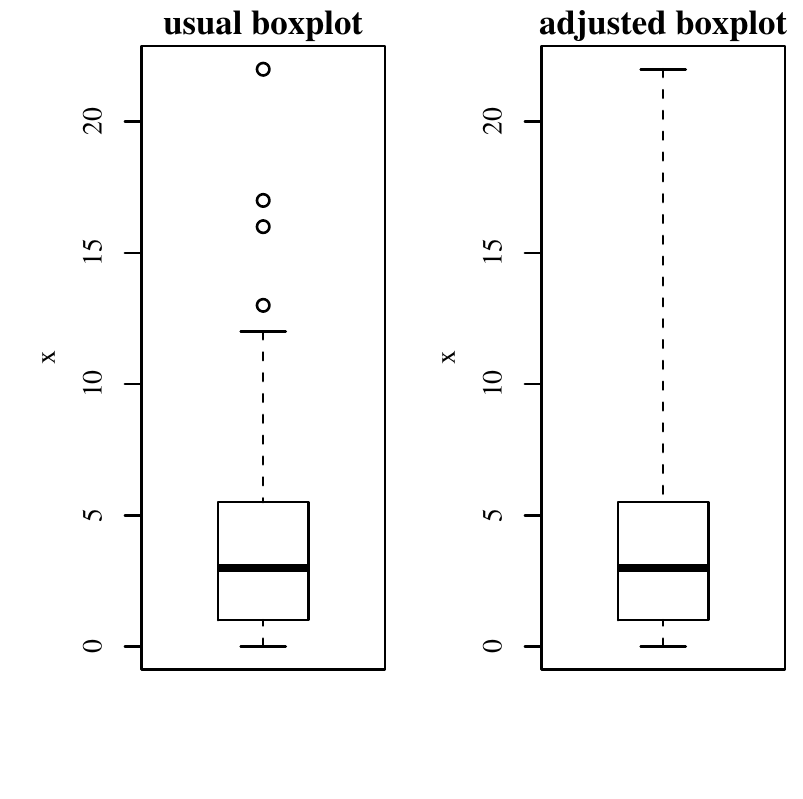}}
		\caption{\small Histogram (left) and boxplots (right) for the computer breaks data.}
		\label{fig:ex2}
	\end{figure}

	Table~\ref{tab:ex2} presents the maximum likelihood estimates, computed by the BFGS method, and standard errors (SEs) for the $\textrm{LS}_{\rm d}$ models parameters. Moreover, the $p$-values of the $\chi^2$ and Cramer-Von Mises (CVM) statistics, and the 
	the Akaike (AIC) and Bayesian information (BIC) criteria, are also reported. The results of Table \ref{tab:ex2} reveal that 
	the discrete BS model provides the best adjustment compared to other models based on the values of AIC and BIC.

	\begin{table}[H]
		\caption{Maximum likelihood estimates (with SE in parentheses) and model selection measures for fit to the computer breaks data.}
		\label{tab:ex2}
		\scalefont{0.81}
		\centering
		\begin{tabular}{lcccccccc}
			\hline
			\multirow{2}{*}{Model}&\multicolumn{3}{c}{Estimates}&&\multicolumn{2}{c}{$p$-value}&\multirow{2}{*}{AIC}&\multirow{2}{*}{BIC}\\
			\cline{2-4}
			\cline{6-7}
			&$\widehat{\lambda}\ (\text{SE})$&$\widehat{\phi}\ (\text{SE})$&$\widehat{\xi}$ ($\widehat{\bm\xi}$)&&$\chi^2$&CMV&&\\
			\hline
			LN&3.2280 (0.2526)&0.7541 (0.1048)&$-$&&0.7841&0.6959&643.5141&652.0702\\
			L$t$&3.2653 (0.2574)&0.7065 (0.1026)&20.0&&0.8306&0.7762&644.2248&652.7809\\
			LCN&3.2283 (0.2526)&0.6858 (0.0953)&(0.9, 0.9)&&0.7996&0.6967&645.5205&656.9287\\
			LPE&3.1624 (0.2770)&1.0176 (0.0555)&-0.2&&0.7096&0.5151&642.8785&651.4346\\
			BS&3.1704 (0.0589)&$*$&0.9&&0.6007&0.5283&640.6061&646.3102\\
			EBS&3.1436 (0.2090)&2.9392 (0.0438)&1.1&&0.6517&0.4700&642.4045&650.9605\\
			BS$t$&3.1803 (0.3008)&$*$&(0.9, 20.0)&&0.7966&0.5799&642.9026&651.4587\\
			EBS$t$&3.1406 (0.2193)&2.0820 (0.0818)&(1.3, 20.0)&&0.6492&0.4364&644.5534&655.9615\\
			\hline
			\multicolumn{9}{l}{{\scriptsize $*$ indicates that $\phi=4$ (fixed)}.}\\
		\end{tabular}
	\end{table}

\end{Example}

\begin{Example}
	
	The second data set refers to the number of physiotherapy sessions until a patient's chronic back pain is reduced or alleviated; see Table \ref{t:data2}. The patients were submitted to electric currents and the study was developed by the School of Physiotherapy Clinics of City University of Sao Paulo (UNICID), Sao Paulo, Brazil; see~\cite{SILVAETAL2017}. Observations were considered censored to the right when patients did not report pain reduction or relief after $12$ treatment sessions, or if they had been lost to follow-up. Such as in \cite{vns:19}, the variable of interest is defined as $T = X - 1$, $t = 0,1,2,3,\ldots$, where $t = 0$ denotes a patient who presented pain relief in the first session performed. Descriptive statistics for the pain relief data are the following: $100$(sample size), $0$(minimum), $11$(maximum), $0$(median), $0.98$(mean), $1.933$(standard deviation), $197.258$(coefficient of variation), $2.802$(coefficient of skewness) and $9.015$(coefficient of kurtosis). These statistics values indicate the positive skewness and a high degree of kurtosis. Figure \ref{fig:ex3} shows the histogram and the fitted survival function by the Kaplan-Meier (KM) method.

	\begin{table}[!ht]
		\footnotesize
		\centering
		\caption{Number of sessions until a patient's chronic back pain is reduced or alleviated.}\label{t:data2}
		\begin{tabular}{ccccccccccccccccccccccc}
			\hline
			Sessions  & $T$  & \# at risk & \# of events & censoring indicator \\
			\hline
			1                   & 0    &  100   & 64      & 0        \\
			2                   & 1    &  36    & 16      & 0        \\
			3                   & 2    &  20    & 5       & 1        \\
			4                   & 3    &  14    & 4       & 0        \\
			5                   & 4    &  10    & 4       & 0        \\
			6                   & 5    &  6     & 3       & 0        \\
			7                   & 6    &  3     & 0       & 0        \\
			8                   & 7    &  3     & 1       & 0        \\
			9                   & 8    &  2     & 0       & 0        \\
			10                  & 9    &  2     & 1       & 0        \\
			11                  & 10   &  1     & 0       & 0        \\
			12                  & 11   &  1     & 0       & 1        \\
			\hline
		\end{tabular}
	\end{table}

	\begin{figure}[!ht]
		\centering
		{\includegraphics[height=7.2cm,width=7.2cm]{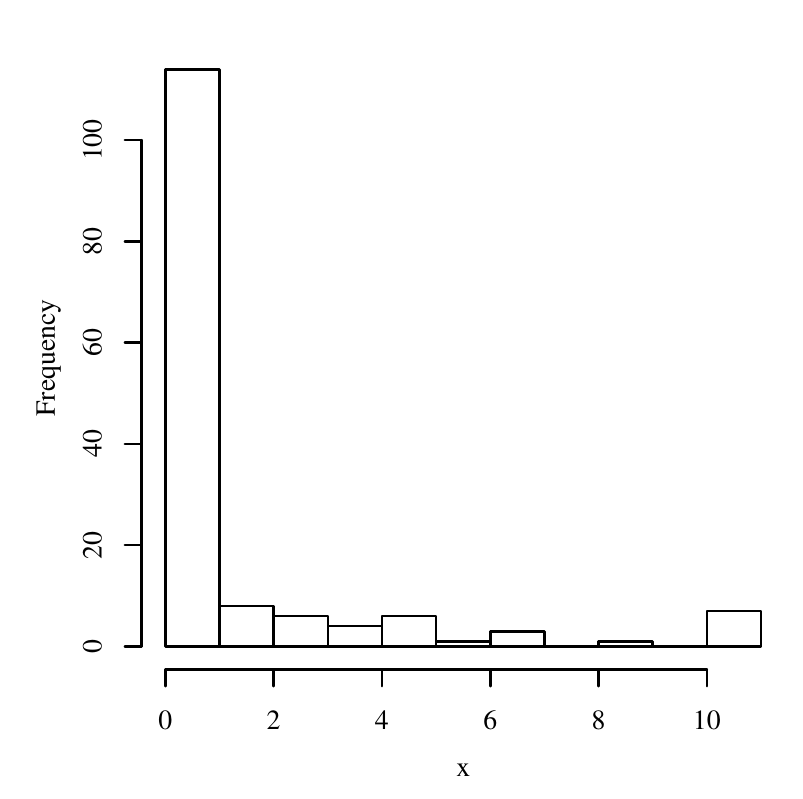}}\hspace{-0.25cm}
		{\includegraphics[height=7.2cm,width=7.2cm]{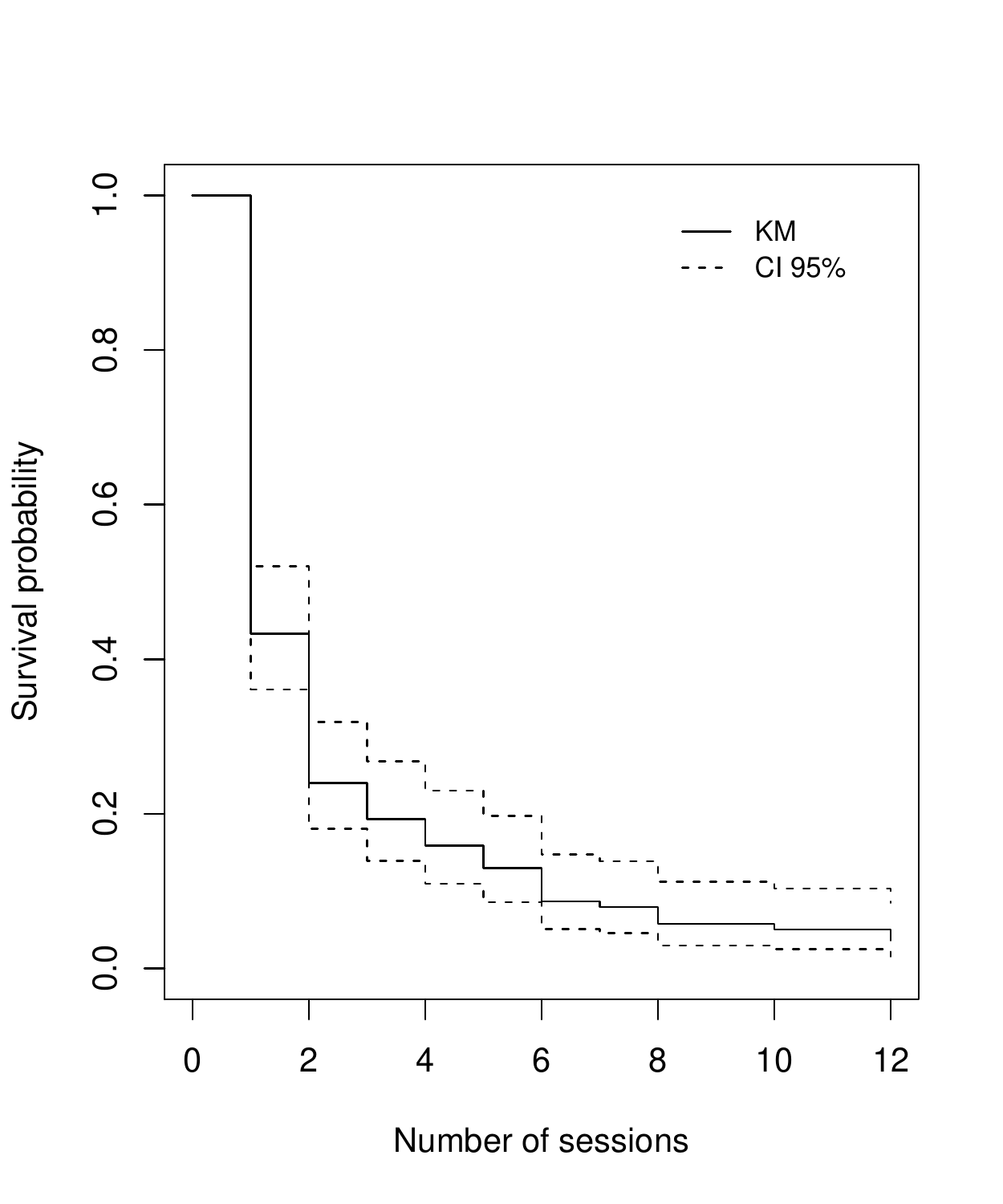}}
		\caption{\small Histogram (left) and KM (right) for the  pain relief data.}
		\label{fig:ex3}
	\end{figure}

	The maximum likelihood estimates of the discrete log-symmetric distribution parameters, along with AIC and BIC criteria are reported in Table \ref{t:ml-censo}. We note that the log-Student-$t$ model provides better adjustment compared to the other models based on the values of AIC and BIC. Table \ref{t:ml-censo} and \ref{t:ci_km} present the fitted survival functions obtained by the KM and the discrete log-symmetric models. These results suggest that (extended) Birnbaum-Saunders and log-normal models yield the best fits to the pain relief data.

	\begin{table}[H]
		\caption{Maximum likelihood estimates and model selection measures for fit to the pain relief data.}
		\label{t:ml-censo}
		\scalefont{0.82}
		\centering
		\begin{tabular}{llcc}
			\hline
			Discrete distribution&Estimates (SE)&AIC&BIC\\
			\hline
			\multirow{2}{*}{Log-normal}&$\hat{\lambda}$=2.3229 (0.1329)&540.0872&549.1191\\
			&$\hat{\phi}$=0.462 (0.0571)&&\\
			\multirow{2}{*}{Log-Student-$t$}&$\hat{\lambda}$=1.8745 (0.1250)&513.1153&522.1472\\
			&$\hat{\phi}$=0.122 (0.0538)&&\\
			&$\hat{\zeta}$=2&&\\
			\multirow{2}{*}{Log-Power-Exponential}&$\hat{\lambda}$=2.0046 (0.0999)&528.1035&537.1354\\
			&$\hat{\phi}$=0.1713 (0.0259)&&\\
			&$\hat{\zeta}$=0.5&&\\
			\multirow{2}{*}{Log-Contamined-Normal}&$\hat{\lambda}$=1.8654 (0.1329)&513.3146&525.3571\\
			&$\hat{\phi}$=0.1018 (0.0519)&&\\
			&$\hat{\zeta}$=(0.37;0.10)&&\\
			\multirow{2}{*}{Birnbaum-Saunders}&$\hat{\lambda}$=2.4767 (0.8043)&543.6507&549.6719\\
			&$\hat{\zeta}$=0.7&&\\
			\multirow{2}{*}{Extended Birnbaum-Saunders}&$\hat{\lambda}$=2.3263 (0.156)&540.2164&549.2483\\
			&$\hat{\phi}$=184.8547 (0.0701)&&\\
			&$\hat{\zeta}$=0.1&&\\
			\multirow{2}{*}{Birnbaum-Saunders-$t$}&$\hat{\lambda}$=1.8966 (0.2047)&515.0695&524.1014\\
			&$\hat{\zeta}$=(0.4;2.0)&&\\
			\multirow{2}{*}{Extended Birnbaum-Saunders-$t$}&$\hat{\lambda}$=1.8751 (0.1477)&515.1634&527.2060\\
			&$\hat{\phi}$=49.0515 (0.1995)&&\\
			&$\hat{\zeta}$=(0.1;2.0)&&\\
			\hline
		\end{tabular}
	\end{table}

	\begin{table}[H]
		\caption{Estimates of the survival function via KM and discrete log-symmetric distributions.}
		\label{t:ci_km}
		\scalefont{0.82}
		\centering
		\begin{tabular}{lcccccccccc}
			\hline
			$x$&KM&LN&L-$t$&LPE&LCN&BS&EBS&BS-$t$&EBS-$t$\\
			\hline
			0&1&0.8925&0.8931&0.8698&0.8848&0.9100&0.8930&0.8774&0.8930\\
			1&0.4333&0.5871&0.4350&0.5018&0.4354&0.6202&0.5880&0.4533&0.4354\\
			2&0.2400&0.3533&0.1552&0.2412&0.1610&0.3919&0.3541&0.1835&0.1557\\
			3&0.1933&0.2120&0.0811&0.1315&0.0885&0.2447&0.2126&0.0982&0.0812\\
			4&0.1588&0.1297&0.0534&0.0791&0.0614&0.1528&0.1301&0.0640&0.0534\\
			5&0.1299&0.0813&0.0398&0.0511&0.0458&0.0958&0.0815&0.0466&0.0396\\
			6&0.0866&0.0523&0.0318&0.0348&0.0352&0.0603&0.0524&0.0364&0.0316\\
			7&0.0794&0.0344&0.0267&0.0247&0.0276&0.0381&0.0344&0.0297&0.0265\\
			8&0.0577&0.0232&0.0231&0.0182&0.0220&0.0242&0.0231&0.0250&0.0228\\
			10&0.0505&0.0111&0.0184&0.0106&0.0145&0.0098&0.0110&0.0190&0.0181\\
			12&0.0361&0.0056&0.0155&0.0067&0.0101&0.0040&0.0056&0.0153&0.0152\\
			\hline
		\end{tabular}
	\end{table}
	
	\begin{figure}[H]
		\centering
		\includegraphics[scale=0.55]{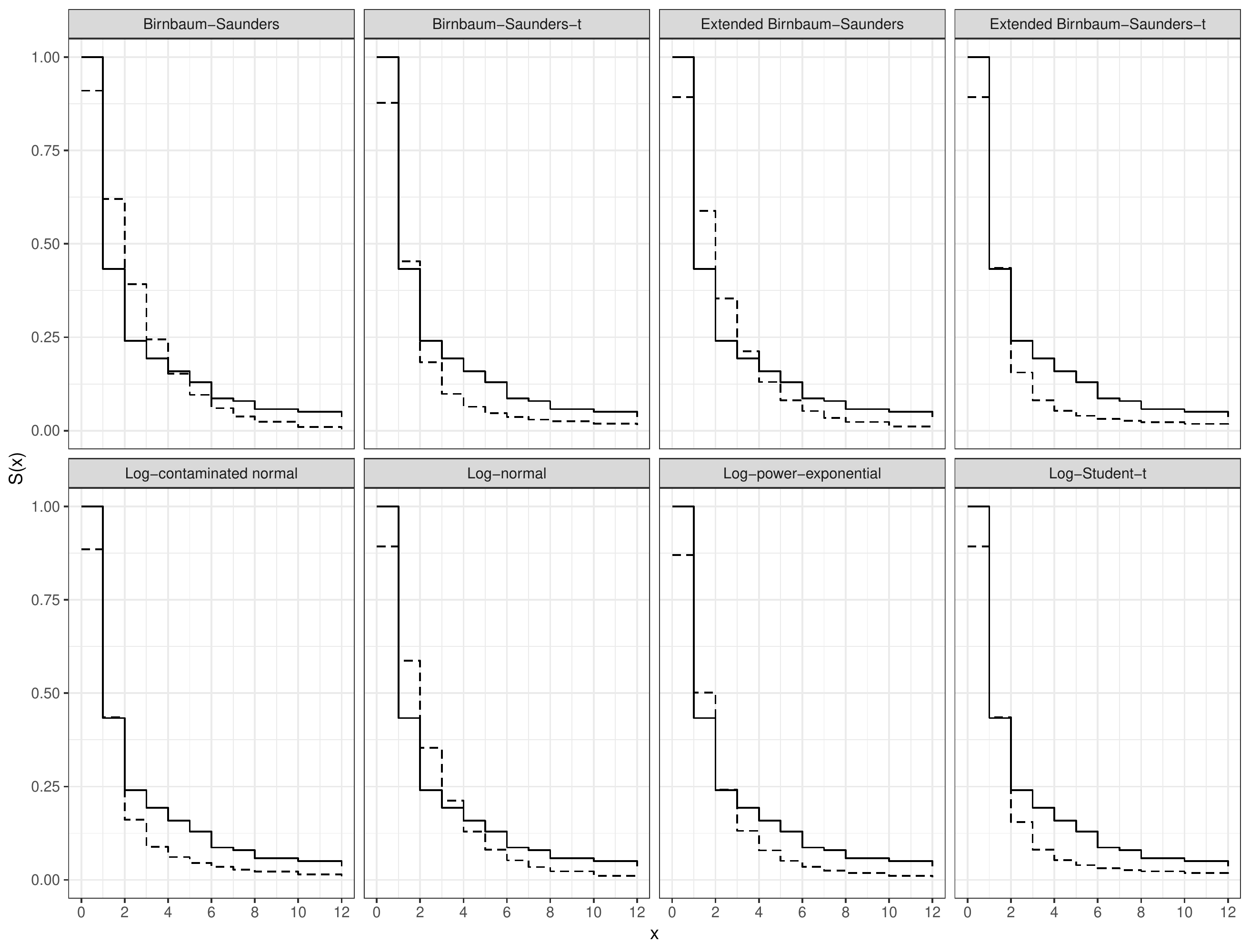}\setlength{\belowcaptionskip}{-8pt} 
		\caption{Estimation of the survival function using the KM (solid) and discrete log-symmetric distributions (dashed) with the pain relief data.}
		\label{fno-}
	\end{figure}

\end{Example}


\newpage
\noindent
\\[0,1cm]
\noindent
\section{Concluding remarks}\label{sec:06}
We have proposed a new class of distributions to deal with cases where the data are discrete, asymmetric and nonnegative. The proposed approach is a discrete version  of the family of continuous log-symmetric distributions. We have considered estimation about the model parameters based on the maximum likelihood method with censored and uncensored data. A Monte Carlo simulation study was carried out to evaluate the behavior of the maximum likelihood estimators. We have applied the proposed models to two real-world data sets. In general, the results have shown that the proposed discrete family proved to be an useful model for discrete data. As part of future research, it is of interest to discuss regression models as well as multivariate extensions. Moreover, time series models based on the proposed class may be of interest. Work on these issues is currently in progress and we hope to report some findings in future papers.
\noindent
\\[0,1cm]
\noindent

\bibliographystyle{apalike}

\begin{thebibliography}{}

\bibitem[Balakrishnan et~al., 2017]{Balakrishnan2017}
Balakrishnan, N., Saulo, H., Bourguignon, M., and Zhu, X. (2017).
\newblock On moment-type estimators for a class of log-symmetric distributions.
\newblock {\em Computational Statistics}, 32(4):1339--1355.

\bibitem[Hinkley, 1975]{hinkley:75}
Hinkley, D.~V. (1975).
\newblock On power transformations to symmetry.
\newblock {\em Biometrika}, 62:101--111.

\bibitem[Jones, 2008]{j:08}
Jones, M.~C. (2008).
\newblock On reciprocal symmetry.
\newblock {\em Journal of Statistical Planning and Inference}, 138:3039--3043.

\bibitem[Medeiros and Ferrari, 2017]{franciscosilvia2017}
Medeiros, F. M.~C. and Ferrari, S. L.~P. (2017).
\newblock Small-sample testing inference in symmetric and log-symmetric linear
regression models.
\newblock {\em Statistica Neerlandica}, 71:200--224.

\bibitem[Mittelhammer et~al., 2000]{mjm:00}
Mittelhammer, R.~C., Judge, G.~G., and Miller, D.~J. (2000).
\newblock {\em {Econometric Foundations}}.
\newblock Cambridge University Press, New York, US.

\bibitem[Moors, 1988]{moors:88}
Moors, J. J.~A. (1988).
\newblock A quantile alternative for kurtosis.
\newblock {\em The Statistician}, 37:25--32.

\bibitem[{R Core Team}, 2016]{R:15}
{R Core Team} (2016).
\newblock {\em {R: A Language and Environment for Statistical Computing}}.
\newblock R Foundation for Statistical Computing, Vienna, Austria.

\bibitem[Saulo and Le\~ao, 2017]{saulo2017log}
Saulo, H. and Le\~ao, J. (2017).
\newblock On log-symmetric duration models applied to high frequency financial
data.
\newblock {\em Economics Bulletin}, 37:1089--1097.

\bibitem[Silva et~al., 2017]{SILVAETAL2017}
Silva, J.~F., Liebano, R.~E., Corr\^ea, J.~B., Matsushita, R.~Y., and Nakano,
E.~Y. (2017).
\newblock Analysis of the time to relieving pain in patients with chronic
non-specific low back pain via {C}ox proportional hazard model.
\newblock {\em Ci\^encia e Natura}, 39:233--243.

\bibitem[Trenkler, 1995]{trenkler:95}
Trenkler, D. (1995).
\newblock A handbook of small data sets: Hand, d.j., daly, f., lunn, a.d.,
mcconway, k.j. \& ostrowski, e. (1994): Chapman \& hall, london.
\newblock {\em Computational Statistics \& Data Analysis}, 19(1):101--101.

\bibitem[Vanegas and Paula, 2016a]{vp:16a}
Vanegas, L.~H. and Paula, G.~A. (2016a).
\newblock An extension of log-symmetric regression models: {R} codes and
applications.
\newblock {\em Journal of Statistical Simulation and Computation},
86:1709--1735.

\bibitem[Vanegas and Paula, 2016b]{vanegasp:16b}
Vanegas, L.~H. and Paula, G.~A. (2016b).
\newblock {\em ssym: Fitting Semi-Parametric log-Symmetric Regression Models}.
\newblock R package version 1.5.7.

\bibitem[Ventura et~al., 2019]{venturaletal:19}
Ventura, M., Saulo, H., Leiva, V., and Monsueto, S. (2019).
\newblock Log-symmetric regression models: information criteria and application
to movie business and industry data with economic implications.
\newblock {\em Applied Stochastic Models in Business and Industry},
35(4):963--977.

\bibitem[Vila et~al., 2019]{vns:19}
Vila, R., Nakano, E.~Y., and Saulo, H. (2019).
\newblock Theoretical results on the discrete {W}eibull distribution of
{N}akagawa and {O}saki.
\newblock {\em Statistics}, 53(2):339--363.

\bibitem[Zwillinger and Kokoska, 2000]{zwko:00}
Zwillinger, D. and Kokoska, S. (2000).
\newblock {\em {Standard Probability and Statistical Tables and Formula}}.
\newblock Chapman \& Hall, Boca Raton.

\end{thebibliography}

\end{document}